\documentclass[12pt, draftclsnofoot, onecolumn]{IEEEtran}
\usepackage[T1]{fontenc}% optional T1 font encoding

\def\argmax{\mathop{\rm \arg\!\max}}
\def\argmin{\mathop{\rm \arg\!\min}}
\usepackage{graphicx}
\usepackage[noadjust]{cite}
\usepackage{mcite}
\usepackage{amsfonts,helvet}
\usepackage{fancyhdr}
\usepackage{threeparttable}
\usepackage{epsf,epsfig}
\usepackage{amsthm}
\usepackage{amsmath}
\usepackage{siunitx}
\usepackage{amssymb}
\usepackage{dsfont}
\usepackage{subfigure}
\usepackage{color}
\usepackage{enumerate}
\usepackage{cancel}
\usepackage{bbm}
\usepackage{dsfont}
\usepackage[subnum]{cases}
\usepackage{adjustbox}
\usepackage{enumitem}
\usepackage[linesnumbered,ruled,noend]{algorithm2e}
\newtheorem{theorem}{Theorem}

\newtheorem{corollary}{Corollary}

\newtheorem{lemma}{Lemma}

\newtheorem{proposition}{Proposition}
\newtheorem{remark}{Remark}

\newmuskip\pFqmuskip

\def\ba{{\bf a}}

\def\bh{{\bf h}}

\def\bn{{\bf n}}

\def\bq{{\bf q}}
\def\br{{\bf r}}
\def\bs{{\bf s}}

\def\bv{{\bf v}}

\def\bx{{\bf x}}
\def\by{{\bf y}}

\def\bA{{\bf A}}
\def\bB{{\bf B}}

\def\bG{{\bf G}}
\def\bH{{\bf H}}
\def\bI{{\bf I}}

\def\bP{{\bf P}}

\def\bR{{\bf R}}

\def\bW{{\bf W}}

\def\cA{\mbox{$\mathcal{A}$}}
\def\cB{\mbox{$\mathcal{B}$}}
\def\cC{\mbox{$\mathcal{C}$}}

\def\cK{\mbox{$\mathcal{K}$}}
\def\cL{\mbox{$\mathcal{L}$}}

\def\cN{\mbox{$\mathcal{N}$}}

\def\cP{\mbox{$\mathcal{P}$}}

\def\cR{\mbox{$\mathcal{R}$}}
\def\cS{\mbox{$\mathcal{S}$}}

\def\cU{\mbox{$\mathcal{U}$}}
\def\cV{\mbox{$\mathcal{V}$}}

\def\bbE{\mbox{$\mathbb{E}$}}

\usepackage{array}

\makeatletter
\newcommand{\thickhline}{%
    \noalign {\ifnum 0=`}\fi \hrule height 1pt
    \futurelet \reserved@a \@xhline
}
\newcolumntype{"}{@{\hskip\tabcolsep\vrule width 1pt\hskip\tabcolsep}}
\makeatother

\IEEEoverridecommandlockouts

\title{\huge User Scheduling for Millimeter Wave Hybrid Beamforming Systems with Low-Resolution ADCs}
\author{
Jinseok Choi, Gilwon Lee, and Brian L. Evans \thanks{
J. Choi and B. L. Evans are with the Wireless Networking and Communication Group (WNCG), Dept. of Electrical and Computer Engineering, The University of Texas at Austin, Austin, TX 78701. (e-mail: \{jinseokchoi89@, bevans@ece.\}utexas.edu). G. Lee is with Intel Corporation, Santa Clara, CA 95054. (e-mail: gilwon.lee@intel.com). 
J. Choi and B. L. Evans were supported by gift funding from Huawei Technologies.
A preliminary version of this work was presented in IEEE ICC 2018 [4].
}
}
\begin{document}
\maketitle

%%%%%%%%%%%%%%%%%%%%%%%%%%%%
\begin{abstract}
	%VERSION 1
We investigate uplink user scheduling for millimeter wave (mmWave) hybrid analog/digital beamforming systems with low-resolution analog-to-digital converters (ADCs). Deriving new scheduling criteria for the mmWave systems, we show that the channel structure in the beamspace, in addition to the channel magnitude and orthogonality, plays a key role in maximizing the achievable rates of scheduled users due to quantization error. The criteria show that to maximize the achievable rate for a given channel gain, the channels of the scheduled users need to have $(1)$ as many propagation paths as possible with unique angle-of-arrivals (AoAs) and $(2)$ even power distribution in the beamspace. Leveraging the derived criteria, we propose an efficient scheduling algorithm for mmWave zero-forcing receivers with low-resolution ADCs. We further propose a chordal distance-based scheduling algorithm that exploits only the AoA knowledge and analyze the performance by deriving ergodic rates in closed form. Based on the derived rates, we show that the beamspace channel leakage resulting from phase offsets between AoAs and quantized angles of analog combiners can lead to sum rate gain by reducing quantization error compared to the channel without leakage. Simulation results validate the sum rate performance of the proposed algorithms and derived ergodic rate expressions.
\end{abstract}

%%%%%%%%%%%%%%%%%%%%%%%%%%%%%%%%%%%%%%%%%%%%%%%%%%%%%%%%

%%%%%%%%%%%%%%%%%%%%%%%%%%%%%%%%%%%%%%%%%%%%%%%%%%%%%%%
\section{Introduction}
\label{sec:intro}
%%%%%%%%%%%%%%%%%%%%%%%%%%%%%%%%%%%%%%%%%%%%%%%%%%%%%%%

%Unlike the traditional MIMO communication that operates sub-3 GHz with a small number of antennas, 

% Millimerter-wave communicatoin introduction
Millimeter wave wireless communication has emerged as a promising technology for next generation cellular systems \cite{pi2011introduction}.
The advantages of remarkably wide bandwidth in mmWave frequencies ranging from
$30-300$ GHz can be exploited to meet ever increasing capacity requirements of wireless communication network.
%have encouraged wireless researchers to perform comprehensive studies to resolve practical challenges in the realization of mmWave communications \cite{niu2015survey, heath2016overview}. 
To compensate for the large path loss of mmWave channels, large antenna arrays are likely to be deployed into tranceivers with very small antenna spacing owing to the small wavelength.
Due to a large signal bandwidth,
% and a high number of bits/sample in mmWave systems, 
high-resolution ADCs coupled with large antenna arrays demand significant power consumption in the receiver, and the power consumption of ADCs scales exponentially in the number of quantization bits.
Therefore, employing low-resolution ADCs has been proposed as a natural solution, and extensive research has been conducted in such systems for mmWave communications \cite{orhan2015low, heath2016overview}.
In this regard, as an extension of our work \cite{choi2017user}, we also investigate low-resolution ADC systems by focusing on user scheduling.
% to provide different scheduling criteria that newly arise from coarse quantization. 
% to take into account the influence of quantization error. 
%Millimeter wave communication has drawn extensive attention as a promising technology for 5G cellular systems \cite{pi2011introduction,andrews2014will,boccardi2014five}, and evinced its feasibility \cite{rappaport2013millimeter}.
%The advantages of remarkably wide bandwidth have encouraged wireless researchers to perform comprehensive studies to resolve practical challenges in the realization of mmWave communications \cite{niu2015survey, heath2016overview}. 
%Due to a large signal bandwidth and a high number of bits/sample in mmWave communications, high-resolution ADCs coupled with large antenna arrays demand significant power consumption in the receiver.
%Consequently, receiver architectures with low-resolution ADCs \cite{fan2015uplink} have been of interest in recent years.

%%%%%%%%%%%%%%%%%%%%%%%%%%%%%  
\subsection{Prior Work}
%%%%%%%%%%%%%%%%%%%%%%%%%%%%%  

% LOW-RESOLUTION ADC
As an effort to realize low-resolution ADC systems, essential wireless communication techniques such as channel estimation and detection have been developed in low-resolution ADC systems \cite{mo2014channel, choi2016near, li2017channel, wen2016bayes, wang2014multiuser, wang2015multiuser}.
For the 1-bit ADC system which is the extreme case of low-resolution ADCs, compressive sensing \cite{mo2014channel}, maximum-likelihood \cite{choi2016near}, and Bussgang decomposition-based techniques \cite{li2017channel} were employed for channel estimation.
Compressive sensing-based channel estimators were also developed for the systems with low-resolution ADCs \cite{wen2016bayes}, and achieved comparable estimation accuracy to that of infinite-bit ADC systems at low and medium signal-to-noise ratio (SNR).
% Unified frameworks for channel estimation and symbol detection were developed for 1-bit ADC systems \cite{choi2016near} and low-resolution ADC systems \cite{wen2016bayes}.
Achieving higher detection accuracy than a minimum mean squared error (MMSE) estimator, message passing de-quantization-based detectors were proposed in 1-bit ADC \cite{wang2014multiuser} and low-resolution ADC systems \cite{wang2015multiuser}.

% LOW-RESOLUTION ADC + HYBRID MIMO 
In recent years, low-resolution ADC systems with hybrid analog/digital beamforming have been investigated to take advantage of both the reduced number of ADC bits and radio frequency (RF) chains \cite{mo2017hybrid, choi2017resolution, choi2017adc, sung2018narrowband}.
% Achievable rates in the hybrid MIMO systems with low-resolution ADCs were characterized in \cite{mo2017hybrid}.
It was shown in \cite{mo2017hybrid} that the hybrid beamforming systems with low-resolution ADCs achieve comparable rate to that of infinite-bit ADC systems, providing better energy-rate trade-off compared to conventional hybrid multiple-input multiple-output (MIMO) systems and low-resolution ADC systems.
To further increase spectral and energy efficiency of mmWave receivers, deploying adaptive-resolution ADCs in hybrid MIMO systems was proposed with ADC bit-allocation algorithms \cite{choi2017resolution,choi2017adc}.
% Since the proposed bit-allocation algorithm in \cite{choi2017resolution} was developed under a total ADC power constraint, a joint binary-search bit-allocation was proposed in \cite{choi2017adc} to take into account a total receiver power constraint.
%The near optimal number of ADC bits were derived in \cite{choi2017resolution} and improved spectral and energy efficiency by allocating different number of bits to the ADCs depending on channel gains.
Channel estimation techniques were also investigated for hybrid MIMO systems with low-resolution ADCs \cite{sung2018narrowband}.
Understanding the superior spectral and energy efficiency of the architecture, we focus on the hybrid MIMO receiver with low-resolution ADCs to solve a user scheduling problem in mmWave communications.

% USER SCHEDULING & QUESTION 
Although user scheduling in multiuser MIMO systems has been extensively studied for more than a decade, it has not been investigated for low-resolution ADC systems.
% Many user scheduling methods were developed under the no quantization system in which the number of quantization bits is considered to be infinite.
One representative method of user scheduling is the semi-orthogonal user selection (SUS) method \cite{yoo2006optimality}. 
%It was shown in \cite{yoo2006optimality} that user scheduling increases an achievable rate by $M\log(\log K)$ where $M$ is the number of transmit antennas and $K$ is the number of users in a cell.
This method selects users in a greedy manner such that the channel vectors of the selected users are nearly orthogonal and have large magnitudes based on the full channel state information (CSI) knowledge of all users at the basestation (BS).  
Another representative approach is the random beamforming (RBF) method \cite{sharif2005capacity} that selects the user who has the maximum signal-to-interference-noise ratio (SINR) for each beam when a set of orthogonal beams are determined a priori at the BS before scheduling. 
%Both SUS and RBF methods were analyzed under the Rayleigh MIMO fading channel model that well captures the channels at lower-frequency bands. 
Similarly, to capture the orthogonality between channels of scheduled users, user scheduling algorithms that adopt chordal distance as a selection measure were proposed in \cite{zhou2011chordal,ko2012multiuser}.
%Since chordal distance can measure the separation between two subspaces, the proposed methods successfully captured the orthogonality between channels.

{\color{black} Unlike the user scheduling methods that have been studied under the Rayleigh fading channel model by assuming rich scattering \cite{yoo2006optimality, sharif2005capacity, liu2010expected}, different approaches have investigated user scheduling under the channels with poor scattering such as mmWave channels \cite{rajashekar2017user, lee2016randomly,lee2016performance}.
In \cite{rajashekar2017user}, user scheduling algorithms were proposed for mmWave communications by leveraging the knowledge of channel gain and angle of departure. 
In addition, the achievable sum rate was quantified for the BS which employs an iterative matrix decomposition
based hybrid beamforming scheme proposed in \cite{rajashekar2017iterative}.}
The RBF method was analyzed in both the uniform random single path \cite{lee2016randomly} and multi-path channel models \cite{lee2016performance}.
% that capture the sparsity of mmWave channels.
%It was shown that there is significant performance difference between the Rayleigh MIMO fading and the uniform random multi-path channel models.
%For millimeter wave communications, user scheduling was investigated in \cite{lee2016randomly} by considering a random beamforming (RBF) method \cite{sharif2005capacity} for uniform random single path channels.
%Then, the analysis of user scheduling in mmWave channels was extended to the uniform random multi-path case, and scheduling algorithms were developed based on a beam selection approach \cite{lee2016performance}.
By exploiting the sparse nature of mmWave channels, beam aggregation-based scheduling and fairness-aware scheduling algorithms were developed in \cite{lee2016performance}. 
Although the user scheduling algorithms were proposed for mmWave communications, 
% outperformed the RBF algorithm in \cite{sharif2005capacity} for mmWave MIMO systems and achieved the sum rate that is comparable to the optimal MMSE case, 
they still focused on user scheduling without quantization error. Consequently, user scheduling in mmWave systems with low-resolution ADCs remains questionable.
\subsection{Contributions}
%%%%%%%%%%%%%%%%%%%%%%%%%%%%

% CHANNEL STRUCTURAL SCHEDULING CRITERIA
{\color{black} 
In this paper, we investigate uplink user scheduling for mmWave hybrid MIMO zero-forcing receivers with low-resolution ADCs. 
Noting that non-negligible quantization error can be a primary bottleneck for attaining scheduling gain in the low-resolution ADC system, we provide following contributions:
\begin{itemize}[leftmargin=*]
    \item We derive user scheduling criteria to maximize the scheduling gain by finding the best tradeoff between channel gains and corresponding quantization noise.
    %We provide the geographical interpretation of the criteria by a
    Adopting the virtual channel model \cite{sayeed2002deconstructing}, the criteria can be interpreted as follows: for a given channel gain, $(i)$ unique AoAs of each scheduled user and $(ii)$ equal power spread across the beamspace complex gains within each user maximize sum rate.
    % with equal channel gains for all users.
    Accordingly, the derived scheduling criteria reveal that the channel structure in the beamspace, in addition to the channel magnitude and orthogonality, plays a key role in maximizing sum rate under coarse quantization. 
    % SCHEDULING ALGORITHM WITH FULL CSI
    \item Leveraging the derived criteria, we propose an efficient scheduling algorithm for hybrid low-resolution ADC systems. 
    The proposed algorithm combines semi-orthogonal user filtering \cite{yoo2006optimality} and non-overlap filtering of dominant beams \cite{lee2016performance} to enforce orthogonality among scheduled users and to reduce quantization error.
    Using an approximated SINR as a scheduling measure, the algorithm captures the trade-off between channel gain and corresponding quantization error, and reduces computational complexity by avoiding matrix inversion.
    %which achieves sub-optimal sum rate performance with low complexity compared to the greedy max-sum rate approach
    % A greedy sum rate maximization algorithm is also proposed to provide a sub-optimal reference performance at the expense of high computational complexity.
    % SCHEDULING ALGORITHM WITH PARTIAL CSI
    \item Considering the difficulty of acquiring instantaneous full CSI, we further propose a chordal distance-based scheduling algorithm which only requires AoAs of mmWave channels, known as slowly-varying channel characteristics \cite{park2017spatial}.  
    % A reasonable alternative to instantaneous full CSI is to exploit slowly-varying channel characteristics, in particular, AoAs of mmWave channels \cite{park2017spatial}.
    Unlike the previously developed chordal distance-based algorithms \cite{zhou2011chordal,ko2012multiuser} that use full CSI and adopt a simple greedy structure which requires prohibitively high complexity, our proposed algorithm exploits only the AoA information of mmWave channels and reduces the complexity by filtering a user candidate set.
    % which leads to the smaller number of remaining candidate users and simpler distance calculation.
    %Moreover, once the favorable candidate users are obtained through filtering, among the users with the most AoAs in the angles that the RF chains see.
    %Moreover, the proposed algorithm computes the chordal distance between only the channel of previously scheduled user and that of filtered candidate users while the previous work computes the distance between the concatenated channel matrix of entire scheduled users and channels of candidate users.
    % ERGODIC RATE ANALYSIS
    \item To analyze the performance of the chordal distance-based algorithm, we derive closed-form sum rates for two channel scenarios: (1) AoAs exactly align with quantized angles of analog combiners and (2) arbitrary AoAs produce phase offsets from the quantized angles, which results in channel leakage.
    % in the beam domain.
    For the first scenario, we derive an ergodic rate as the sum of the ergodic rate with no quantization and the rate loss due to quantization.
    Accordingly, the derived rate provides the expected ergodic rate loss due to quantization in closed form. 
    For the second scenario, an approximated lower bound of the ergodic rate is derived in closed form.
    % and offers a performance guideline for a more realistic channel scenario.
    We observe that the two channel scenarios result in different sum rates as a consequence of coarse quantization, and the channel leakage provides sum rate gain by reducing quantization error, which challenges the conventional negative understanding towards channel leakage.
    %The derived ergodic rate expressions shows the convergence in the high-resolution ADC regime, implying that the two channel scenarios can result in different sum rates due to coarse quantization.
    % Such intuition is numerically verified and shows that the channel leakage provides sum rate gain by reducing quantization error, which challenges the conventional negative understanding towards channel leakage.
    % while the exact AoA alignment scenario is often considered to be more favorable in mmWave communications
%The proposed scheduling algorithm which exploits full CSI outperforms conventional user scheduling methods in \cite{yoo2006optimality,lee2016performance}, and the performance gap increases as transmit power increases, the number of RF chains increases, and/or quantization resolution decreases.
%The chordal distance-based scheduling algorithm that uses the AoA information of channels achieves large sum rate gain compared to random scheduling, closing the performance gap from the full CSI based algorithms in \cite{yoo2006optimality,lee2016performance}.
%It is also verified that the channel leakage provides an increase in sum rate by reducing quantization error compared to the channels with no leakage.
\end{itemize}
% NUMERICAL STUDY & CHANNEL LEAKAGE EFFECT
% Simulation results demonstrate that the proposed algorithm outperforms the conventional algorithms developed under the no quantization assumption in ergodic rate and the chordal distance-based algorithm also improves the sum rate compared to random scheduling.
Simulation results demonstrate the superior ergodic sum rate performance of the proposed algorithms and validate the analysis and intuition obtained in this paper.
}

{\it Notation}: $\bf{A}$ is a matrix and $\bf{a}$ is a column vector. 
$\mathbf{A}^{H}$ and $\mathbf{A}^\intercal$  denote conjugate transpose and transpose. 
$[{\bf A}]_{i,:}$ and $ \mathbf{a}_i$ indicate the $i$th row and column vector of $\bf A$. 
We denote $a_{i,j}$ as the $\{i,j\}$th element of $\bf A$ and $a_{i}$ as the $i$th element of $\bf a$. 
$\mathcal{CN}(\mu, \sigma^2)$ is a complex Gaussian distribution with mean $\mu$ and variance $\sigma^2$. 
$\mathbb{E}[\cdot]$ and ${\rm Var}[\cdot]$ represent expectation and variance operator, respectively.
The cross-correlation matrix is denoted as ${\bf R}_{\bf xy} = \mathbb{E}[{\bf x}{\bf y}^H]$.
The diagonal matrix $\rm diag(\bf A)$ has $\{a_{i,i}\}$ at its $i$th diagonal entry, and $\rm diag (\bf a)$ or $\rm diag(\bf a^\intercal)$ has $\{a_i\}$ at its $i$th diagonal entry. 
${\bf I}_N$ denotes an $N \times N$ identity matrix and $\|\bf A\|$ represents $L_2$ norm.
$|\cdot|$ indicates an absolute value for a complex value or denotes cardinality of a set.  
${\rm tr}(\cdot)$ is a trace operator.

%%%%%%%%%%%%%%%%%%%%%%%%%%%%%%%%%%%%%%%%%%%%%%%%%%%%%%%%
\section{System Model}
\label{sec:sys_model}
%%%%%%%%%%%%%%%%%%%%%%%%%%%%%%%%%%%%%%%%%%%%%%%%%%%%%%%%

% FIGURE %%%%%%%%%%%%%%%%%%%%%%%%%%%%%%%%%%%%%
\begin{figure}[!t]\centering
\includegraphics[scale = .3]{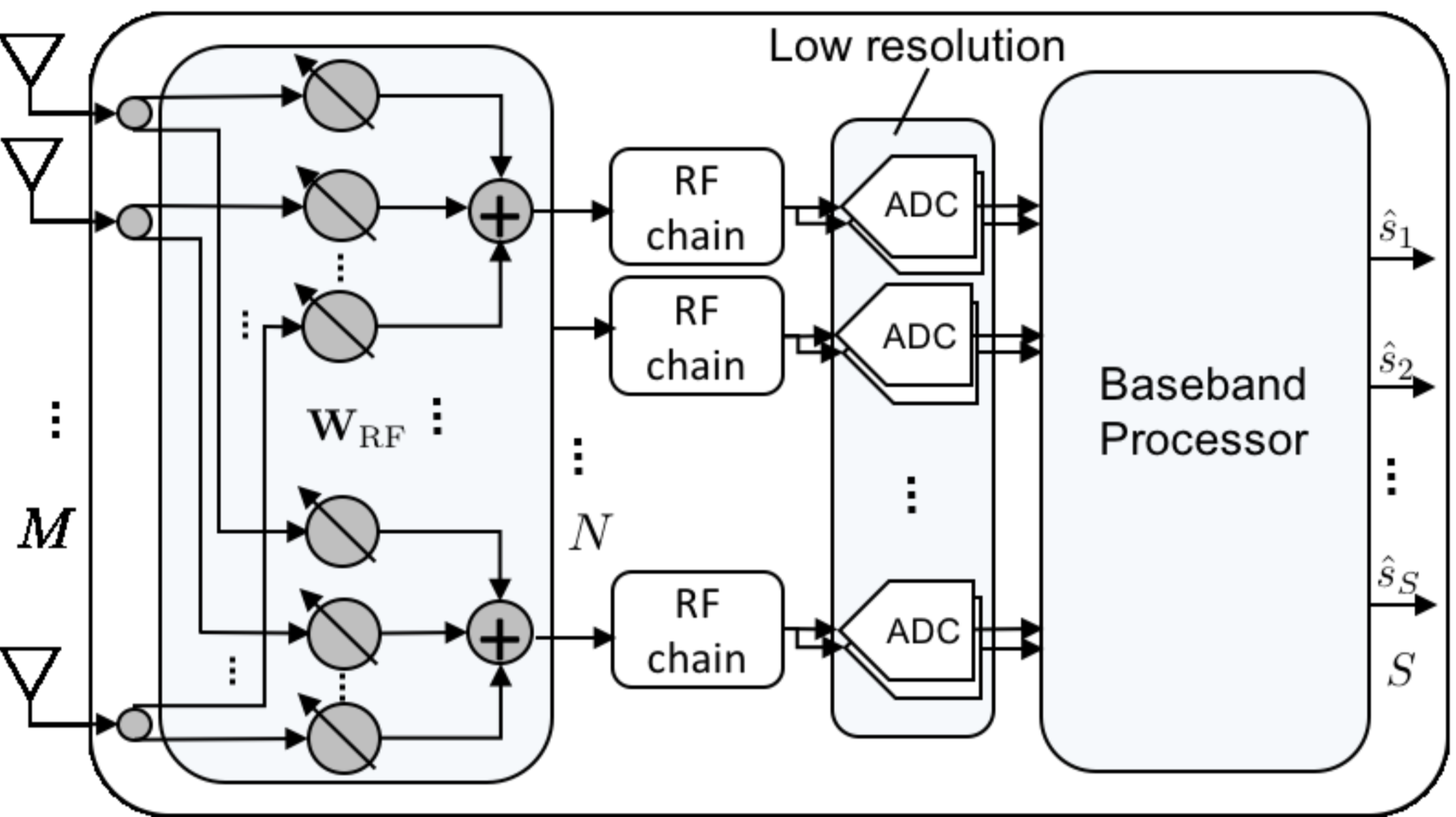}
\caption{A receiver architecture with large antenna arrays and analog combiners ${\bf W}_{\rm RF}$, followed by low-resolution ADCs. } 
\label{fig:system}
\vspace{-1em}
\end{figure}
%%%%%%%%%%%%%%%%%%%%%%%%%%%%%%%%%%%%%%%%%%%%%%

%%%%%%%%%%%%%%%%%%%%%%%%
\subsection{Signal and Channel Models}
%%%%%%%%%%%%%%%%%%%%%%%%

We consider a single-cell multiuser MIMO network for uplink communications.
A BS employs a uniform linear array (ULA) of $M$ receive antennas.
Analog combiners are applied at the BS, followed by $N \leq M$ chains as shown in Fig. \ref{fig:system}.
We assume that $K$ single-antenna users are distributed in the cell and the BS schedules $S \leq N$ users to serve among the $K$ users in the cell.
The ADCs are considered to be low-resolution ADCs to reduce the receiver power consumption.
%After analog combining, the received baseband analog signals are processed through $N_{\rm RF}$ RF chains and quantized at ADCs that are connected to the RF chains.

{\color{black}
Focusing on mmWave communications, the channel ${\bf h}_k$ for user $k$ is assumed to be a sum of the contributions of limited scatterers that contribute $L_k$ propagation paths to the channel ${\bf h}_k$  \cite{brady2013beamspace}.
Therefore, the discrete-time narrowband channel of user $k$ can be modeled as \cite{sayeed2002deconstructing}
\begin{align}
	\label{eq:channel_geo}
	\bh_{\gamma,k}=\sqrt{\frac{1}{\gamma_k}}{\bf h}_k = \sqrt{\frac{M}{\gamma_k L_k}}\sum_{\ell = 1}^{L_k}g_{k,\ell} {\bf a}(\phi_{k,\ell})
\end{align}
where $\gamma_k$ denotes the pathloss of user $k$, $g_{k,{\ell}}$ is the complex gain of the $\ell${th} propagation path of user $k$, and ${\bf a}(\phi_{k,{\ell}})$ is the array steering vector of the BS receive antennas corresponding to the azimuth AoA of the $\ell$th path of the $k$th user $\phi_{k,{\ell}} \in [-\pi/2,\pi/2]$. 
We consider that $g_{k,{\ell}}$ is an independent and identically distributed (IID) complex Gaussian random variable as $g_{k,{\ell}} \overset{i.i.d}{\sim} \mathcal{CN}(0, 1)$.
The array steering vector ${\bf a}(\theta)$ for the ULA antennas of the BS is given as
\begin{align}
	\label{eq:ARV}
	% {\bf a}(\theta) = \frac{1}{\sqrt{N_r}}\Big[1,e^{-j \frac{2\pi d}{\lambda}\sin(\theta)},e^{-j \frac{4\pi d}{\lambda}\sin(\theta)},\dots,e^{-j \frac{2(N_r-1)\pi d}{\lambda}\sin(\theta)}\Big]^\intercal 
	{\bf a}(\theta) = \frac{1}{\sqrt{M}}\Big[1,e^{-j \pi\vartheta},e^{-j 2\pi \vartheta},\dots,e^{-j (M-1)\pi\vartheta}\Big]^\intercal 
\end{align}
where $\vartheta = \frac{2d}{\lambda}\sin(\theta)$ is the spatial angle that is related to the physical AoA $\theta$, $d$ denotes the distance between antenna elements, and $\lambda$ represents the signal wave length.
Throughout this paper, we use $\theta$ and $\phi$ to denote the physical angles of analog combiners and physical AoAs of a user channel, respectively.
We also use $\vartheta$ and $\varphi$ to indicate the spatial angles for $\theta$ and $\phi$, respectively.
We assume that $\vartheta$ is a constant value in the range of $[-1,1]$ and $\varphi$ is a uniform random variable $\varphi\sim {\rm Unif}[-1,1]$.

For simplicity, we consider a homogeneous long-term received SNR network\footnote{\color{black} We remark that the proposed scheduling criteria and the proposed algorithms in this paper can also be applicable to a heterogeneous long-term received SNR network.} where a conventional uplink power control compensates for the pathloss and shadowing effect to achieve the same long-term received SNR target for all users in the cell \cite{simonsson2008uplink, tejaswi2013survey}. 
Let $\bx =\bP\bs$ be the transmitted user signals where $\bP ={\rm diag}\{\sqrt{\rho\, \gamma_1},\dots,\sqrt{\rho\,\gamma_S}\}$ is the transmit power matrix and $\bs$ is the $S \times 1$ transmitted symbol vector from $S$ users.
Let $\bH_\gamma = \bH\bB$ represent the $M \times S$ channel matrix where {\color{black} $\bH_\gamma = [\bh_{\gamma,1},\dots,\bh_{\gamma,S}]$ is the channel matrix, $\bH = [\bh_{1},\dots,\bh_{S}]$ is the channel matrix after the uplink power control, and $\bB = {\rm diag}\{\sqrt{1/\gamma_1},\dots,\sqrt{1/\gamma_S}\}$ is the pathloss matrix.}
Then, the received baseband analog signal ${\bf r} \in \mathbb{C}^{M}$ is given as
\begin{align}
	\label{eq:received_analog_signal}
	\br = \bH_\gamma \bx + \bn =\bH\bB \bP\bs + \bn  =  \sqrt{\rho}\bH\bs + \bn
\end{align} 
where we assume  ${\bf s} \sim \mathcal{CN}({\bf 0}, {\bf I}_S)$, and ${\bf n}$ indicates the additive white Gaussian noise (AWGN) vector ${\bf n}  \sim \mathcal{CN}(\mathbf{0}, \mathbf{I}_M)$. Thus, we can regard $\rho$ as the SNR. 
% i.e., $\mathbb{E}\big[s_i\big] = 0$ and ${\rm Var}\big[s_i\big] = 1$,
%We consider that the channel $\bf H$ is known at the BS.
% its elements are constrained to have the equal norm of $1/\sqrt{N_r}$ 

The received analog signals in \eqref{eq:received_analog_signal} are combined via an $M \times N$ analog combiner ${\bf W}_{\rm RF}$.
%In this paper, we assume $N_{\rm RF} = M$ as the BS employs low-resolution ADCs for reducing power consumption.
The combiner ${\bf W}_{\rm RF}$ is implemented using analog phase shifters, and its elements are constrained to have the equal norm of $1/\sqrt{M}$.
%, which satisfies $[{\bf w}_{{\rm RF},i }{\bf w}^H_{{\rm RF},i }]_{\ell,\ell} = 1/M$ where $[\cdot]_{\ell,\ell}$ denotes the $\ell$th diagonal element of a matrix.
After analog combining, \eqref{eq:received_analog_signal} becomes
% EQUATION
\begin{align} 
	\label{eq:y}
	{\bf {y}}  = {\bf W}_{\rm RF}^H {\bf r} = \sqrt{\rho}{\bf  W}_{\rm RF}^H{\bf Hs} +  {\bf W}_{\rm RF}^H {\bf n}.
\end{align}
%where ${\bf n} =  {\bf A}_{\rm RF}^H \tilde{\bf n} \sim \mathcal{CN}(\mathbf{0}, \mathbf{I}_{N_{\rm RF}})$ as $\bf A$ is unitary, and 
%where $\vartheta$ is the spatial angle that is related to the physical AoA as $\vartheta = \frac{d}{\lambda}\sin(\theta)$, $d$ denotes the distance between antenna elements, and $\lambda$ represents the signal wave length.
%$\lambda$ is a signal wave length, and $d$ is the distance between antenna elements.
%is a vector whose $n$th entry is $\frac{1}{\sqrt{N_r}}e^{-j \frac{2\pi (n-1)d}{\lambda}\sin(\theta)}$
% is given by
%\begin{align}
%\nonumber
% {\bf a}(\theta) = \frac{1}{\sqrt{N_r}}\Big[1,e^{-jv_1\sin(\theta)},e^{-j v_2\sin(\theta)},\dots,e^{-j v_{N_r -1}\sin(\theta)}\Big]^\intercal 
%\end{align}
%with $v_n = \frac{2\pi nd}{\lambda}$, 
%where $d$ denotes the distance between antenna elements and $\lambda$ represents the wave length.
Assuming uniformly-spaced spatial angles, 
%i.e., $\vartheta_i = 2(i-1)/M$ for $i = 1,\dots,M$, 
the matrix of array steering vectors
$\mathbf{A}=\big[{\bf a}(\theta_1),\dots,$ ${\bf a}(\theta_{M})\big]$
becomes a unitary discrete Fourier transform (DFT) matrix. 
Noting that the antenna space and beamspace are related through a spatial Fourier transform, we adopt a sub-matrix of the DFT matrix as the analog combiner ${\bf W}_{\rm RF} = \tilde{\bf A}$ \cite{alkhateeb2014channel, choi2017resolution} to project the received signals onto the beamspace, where $\tilde{\bf A}$ consists of $N$ columns of ${\bf A}$.
Through the projection, the BS can exploit the sparsity of the mmWave channels to capture channel gains with the reduced number of RF chains \cite{el2014spatially}.
Using ${\bf W}_{\rm RF} = \tilde{\bf A}$, we rewrite \eqref{eq:y} as
%The received signal after the analog combining \eqref{eq:y} can be rewritten as 
\begin{align}
	\label{eq:combining_output}
	{\bf y} = \sqrt{\rho} \tilde{\bf A}^H{\bf H}{\bf s} + \tilde{\bf A}^H{\bf n} = \sqrt{\rho}\bH_{\rm b}\bs + \bv.
\end{align}
We denote ${\bf H}_{\rm b} = \tilde{\bf A}^H{\bf H}$, which is the projection of the channel matrix onto the beamspace.
Since $\bA$ is a unitary matrix, the projected noise vector $\bv = \tilde{\bf A}^H{\bf n}$ is  distributed as $\cC\cN({\bf 0},\bI_N)$.
}
%In the virtual channel representation \cite{sayeed2002deconstructing}, the beamspace channel has either zero or non-zero elements in which the non-zero elements follows the complex Gaussian distribution.

%%%%%%%%%%%%%%%%%%%%%%%%%%
\subsection{Quantization Model}
\label{subsec:quantization}
%%%%%%%%%%%%%%%%%%%%%%%%%%

%%%%%%%%%%%%%%%%%%%% TABLE %%%%%%%%%%%%%%%%%%% 
%\begin{table}[!t]
%\centering
%\caption{The Values of $\beta$ for Different Quantization Bits $b$ }
%\label{tb:beta}
%\begin{tabular}{ l c c c c c }
%  \thickhline
%$ b$  & 1 & 2 & 3 & 4 & 5\\
%  \hline
% $\beta$   & 0.3634 & 0.1175 & 0.03454 & 0.009497 & 0.002499 \\
%  \thickhline
%\end{tabular}
%\end{table}
%%%%%%%%%%%%%%%%%%%%%%%%%%%%%%%%%%%%%%%%%%%%%%% 

{\color{black} In this subsection, we introduce an additive quantization noise model \cite{fletcher2007robust} which approximates quantization process in a linear form for analytical tractability.}
Such linear approximation of quantization provides reasonable accuracy in low and medium SNR ranges \cite{orhan2015low}.
After processed through the RF chains, each complex sample $y_i$ in \eqref{eq:combining_output} is quantized at the $i$th pair of ADCs, and each ADC quantizes either a real or imaginary component of $y_i$.
%For analytical tractability, we adopt the additive quantization noise model (AQNM) to represent the quantization process in a linear form. 
%The AQNM provides reasonable accuracy in low and medium SNR ranges \cite{orhan2015low}.
The quantized signal $\bf y_{\rm q}$ is \cite{fletcher2007robust}
% EQUATION
\begin{align} 
	\label{eq:yq}
	{\bf y}_{\rm q} & = \mathcal{Q}\bigl({\rm Re}\{\by\}\bigr)  + j\mathcal{Q}\bigl({\rm Im}\{\by\}\bigr) = \alpha \sqrt{\rho} {\bf H}_{\rm b}{\bf s}+ \alpha {\bf v} +{\bf q}
%	\mathbf{y}_{\rm q}&=\mathcal{ Q}(\mathbf{ y}) = \alpha \sqrt{\rho} {\bf H_{\rm b}}{\bf s} +\alpha {\pmb \eta} + {\bf q}
\end{align} 
where $\mathcal{Q}(\cdot)$ is the element-wise quantizer function.
The quantization gain $\alpha$ is defined as $\alpha = 1- \beta$, $\beta = {\mathbb{E}[|y - {y}_{{\rm q}}|^2]}/{\mathbb{E}[|{y}|^2]}$ is a normalized mean squared quantization error, and ${\bf q}$ is the additive quantization noise vector.

For a scalar MMSE quantizer of a Gaussian random variable, $\beta$ can be approximated as $\beta \approx \frac{\pi\sqrt{3}}{2} 2^{-2b}$ for $b > 5$ \cite{gersho2012vector} where $b$ denotes the number of quantization bits for each real and imaginary part of $y$.
The values of $\beta$ for $b \leq 5$ are shown in Table 1 in \cite{choi2017resolution}.
%Since ${\bf A}$ is unitary, the noise  ${\bv} = {\bf A}^H {\bf n}$ is distributed as ${\bv} \sim \mathcal{CN}({\bf 0},{\bf I})$.
{\color{black} Although the quantization error is neither Gaussian nor is its covariance matrix diagonal in an exact nonlinear quantization model, we provide approximations based on \cite{orhan2015low, fletcher2007robust, mezghani2012capacity} as follows:
considering a lower bound of achievable rate, we assume ${\bf q} \sim \mathcal{CN}({\bf 0},{\bf R}_{\bf qq}(\bH_{\rm b}))$ \cite{mezghani2012capacity}.
Since ${\bf q}$ is uncorrelated with $\bf y$ \cite{fletcher2007robust}, the covariance matrix of $\bq$ with $ {\bf H}_{\rm b}$ is given as  \cite{fletcher2007robust,mezghani2012capacity}}
% EQUATION
\begin{align}
	\label{eq:cov}
	\mathbf{R}_{\bf qq}(\bH_{\rm b})= \alpha(1-\alpha)\,{\rm diag}(\rho{\bf H_{\rm b}}{\bf H}_{\rm b}^H + {\mathbf{I}_N}).
\end{align}
%In the following section, we explore a user scheduling problem for the mmWave communication system with low-resolution ADCs described in this section.
In the following section, we investigate a user scheduling problem based on the considered system model.

%%%%%%%%%%%%%%%%%%%%%%%%%%%%%%%%%%%%%%%%%%%%%%%%
\section{User Scheduling}
\label{sec:CSIR}
%%%%%%%%%%%%%%%%%%%%%%%%%%%%%%%%%%%%%%%%%%%%%%%%

%Since a zero-forcing combiner provides the optimal performance in user scheduling as the number of users increases \cite{yoo2006optimality}, 
In this section, we focus on zero-forcing (ZF) combining ${\bf W}_{\rm zf} = {\bf H}_{\rm b}({\bf H}_{\rm b}^H {\bf H}_{\rm b})^{-1}$ at the BS and investigate user scheduling to derive scheduling criteria and propose an algorithm by exploiting the obtained criteria.
To this end, we first consider the case where the effective CSI ${\bf H}_{\rm b}$ is known at the BS and then extend the problem to the case where only the partial CSI is available.
{\color{black} For low-resolution ADC systems, state-of-the-art channel estimation techniques have been developed and have shown remarkable estimation accuracy with few-bit ADCs \cite{rusu2015low,wen2016bayes} or even with one-bit ADCs \cite{mo2014channel,choi2016near,li2017channel}.}
With the ZF combiner $\bW_{\rm zf}$, the quantized signal in \eqref{eq:yq} is given as 
\begin{align}
	\nonumber
	{\bf y}^{\rm zf}_{\rm q}
	& = {\bf W}_{\rm zf}^H {\bf y}_{\rm q} = \alpha \sqrt{\rho}{\bf W}_{\rm zf}^H {\bf H}_{\rm b} {\bf s} + \alpha {\bf W}_{\rm zf}^H {\bv} + {\bf W}_{\rm zf}^H {\bf q}.
\end{align}
Nulling out the inter-user interference,
%we have the received signal for user $k$ as
%\begin{align}
%	\nonumber
%	y_{{\rm q},k}^{\rm zf} 
%	= &\, \alpha \sqrt{\rho} s_k  + \alpha {\bf w}_{{\rm zf},k}^H {\bv} + {\bf w}_{{\rm zf},k}^H {\bf q}.
%%	= &\, \alpha \sqrt{\rho} {\bf w}_{{\rm zf},k}^H {\bf h}_{{\rm b},k} + \alpha \sqrt{\rho} \sum_{i \neq k}^{N_s} {{\bf w}_{{\rm zf},k}^H {\bf h}_{{\rm b},i}s_i} \\
%%	& + \alpha {\bf w}_{{\rm zf},k}^H {\pmb \eta} + {\bf w}_{{\rm zf},k}^H {\bf q}.
%\end{align}
the achievable rate of user $k$ is derived as
\begin{align}
	\label{eq:rate}
	&{r}_k({\bf H}_{\rm b})
	 = \log_2 \left(1+\frac{\alpha^2 \rho}{ {\bf w}_{{\rm zf},k}^H {\bf R}_{{\bf qq}}(\bH_{\rm b}){\bf w}_{{\rm zf},k}+\alpha^2 \|{\bf w}_{{\rm zf},k}\|^2}\right)
%	\\ 
%	&= \log_2 \left(1+\frac{\alpha \rho}{\beta{\bf w}_{{\rm zf},k}^H {\rm diag}\Big(\rho{\bf H}_{\rm b} {\bf H}_{\rm b}^H\Big) {\bf w}_{{\rm zf},k} + \|{\bf w}_{{\rm zf},k}\|^2}\right). 
\end{align}
%where the quantization noise variance ${\bf R}_{\bf qq}$ is defined in \eqref{eq:cov}.
Using the achievable rate with quantization error \eqref{eq:rate}, we formulate a user scheduling problem:
\begin{align}
	\label{eq:p1}
	\mathcal{P}1: \quad \mathcal{R}({\bf H}_{\rm b}(\mathcal{S^\star})) =\max_{\mathcal{S} \subset \{1,\dots,K\}:|\mathcal{S}| \leq S} \sum_{k \in \mathcal{S}} {r}_k({\bf H}_{\rm b}(\mathcal{S}))
\end{align}
where ${\mathcal{S}}$ represents the set of scheduled users, ${\bf H}_{\rm b}(\mathcal{S})$ is the beamspace channel matrix of the users in $\mathcal{S}$, and $\cR({\bf H}_{\rm b}(\mathcal{S}))$ is the sum rate of the scheduled users in $\cS$.
Unlike the user scheduling without quantization, which considers the channel orthogonality and the large channel gains, the user scheduling with the coarse quantization needs to consider an additional condition.
% as discussed in the following remark.
\begin{remark}
	\label{rm:intuition}
	To maximize the achievable rate \eqref{eq:rate}, the aggregated beamspace channel gain at each RF chain $\|[{\bf H}_{\rm b}]_{i,:}\|^2$ needs to be minimized to reduce the quantization noise variance ${\bf R}_{\bf qq}$ in addition to forcing the channel orthogonality $({\bf h}_{{\rm b},k}\! \perp\! {\bf h}_{{\rm b},k'},\ k \neq k')$ and maximizing the beamspace channel gain $\|{\bf h}_{{\rm b},k}\|^2$, which reduces $\|{\bf w}_{{\rm zf},k}\|^2$.
	%	To maximize \eqref{eq:rate}, the aggregated beamspace channel gain at each RF chain $\|[{\bf H}_{\rm b}]_{i,:}\|^2$ needs to be small to reduce the quantization error ${\bf R}_{\bf qq}$ in addition to the necessity for the beamspace channel orthogonality $({\bf h}_{{\rm b},k} \perp {\bf h}_{{\rm b},k'})$ for $k \neq k'$ and large beamspace channel gains $\|{\bf h}_{{\rm b},i}\|^2$ to reduce $\|{\bf w}_{{\rm zf},i}\|^2$, where $[{\bf H}_{\rm b}]_{i,:}$ represents the $i$th row of ${\bf H}_{\rm b}$.
\end{remark}
%\vspace{-1em}
%\begin{remark}
%\label{rm:intuition}
%	There are two major conditions that need to be met in order to maximize the achievable rate for the user selection problem: ($i$) With respect to users, channels need to be orthogonal to each other with large gains to reduce $\|{\bf w}_{{\rm zf},k}\|$.
%	($ii$) With respect to beams, however, beam gains need to be small to reduce each element in ${\rm diag}\Big(\rho{\bf H}_{\rm b} {\bf H}_{\rm b}^H + \frac{1}{1-\alpha}{\bf I}_{N_{\rm RF}}\Big)$.
%\end{remark}
%
%Accordingly, SUS is not appropriate for quantized signals. 
%Based on Remark \ref{rm:intuition}, we analyze the optimal user selection criterion under particular conditions.

%%%%%%%%%%%%%%%%%%%%%%%%%%%%%%%%%
\subsection{Analysis of Scheduling Criteria}
%%%%%%%%%%%%%%%%%%%%%%%%%%%%%%%%%

We derive the scheduling criteria for channels in the beamspace based on the finding in Remark \ref{rm:intuition} to propose an efficient scheduling algorithm that solves $\mathcal{P}1$ in \eqref{eq:p1}.
To focus on key scheduling ingredients besides the channel magnitude,  we consider the case where the magnitude of each user channel is given in the analysis, i.e., $\|{\bf h}_{{\rm b},k}\| = \sqrt{\gamma_k},\ \forall k$ with $\gamma_k >0$.
Given $\|{\bf h}_{{\rm b},k}\| = \sqrt{\gamma_k},\ \forall k$, we reformulate $\mathcal{P}1$ to the problem of finding the optimal channel matrix that maximizes the uplink sum rate to characterize the channel matrix that fully extracts scheduling gains.
\begin{align}
	\label{eq:p2}
	\mathcal{P}2: \quad \mathcal{R}({\bf H}_{\rm b}^\star) =\max_{{\bf H}_{\rm b}\in \mathbb{C}^{N \times S}}\sum_{k = 1}^{S} r_k({\bf H}_{\rm b}), \quad \text{s.t. } \|{\bf h}_{{\rm b},k}\| = \sqrt{\gamma_k} \ \ \forall k.
\end{align}
%Note that the magnitude of the beamspace channel $\|{\bf h}_{{\rm b},k}\| = \sqrt{\gamma}$ in \eqref{eq:p2} is equivalent to $\|{\bf h}_{k}\| = \sqrt{\gamma}$ since the analog combiner ${\bf W}_{\rm RF} = {\bf A}$ is unitary. 

%Analyzing the optimal channel structure ${\bf H}_{\rm b}^\star$, the optimal user selection criterion can be obtained.
%%Selecting optimal $N_s$ users to serve in $N_u$ users, 

To provide geometrical interpretation for the channel matrix analysis, we further adopt the virtual channel representation \cite{sayeed2002deconstructing}, where
% that uses spatial beams in fixed virtual directions. 
each beamspace channel ${\bf h}_{{\rm b},k}$ contains $(N - L_k)$ zeros and $L_k$ complex gains of the $L_k$ channel paths.
We first consider the single user scheduling ($S = 1$) and derive the channel characteristics required to maximize the achievable rate for $\cP 2$.
Then, we utilize the result to derive the scheduling criteria for the multiuser scheduling case. 
%The analysis will show the optimality condition for the channel structure that SUS cannot satisfy.
\begin{lemma}
	\label{lem:rate1}
	For a single user scheduling ($S=1$), scheduling a user who has the following channel characteristics maximizes the uplink achievable rate in $\cP 2$:
	\begin{enumerate}
		\item[(i)] the largest number of channel propagation paths and 
		\item[(ii)] equal power spread across the beamspace complex gains.
	\end{enumerate}
%	$(i)$ the largest number of channel propagation paths and $(ii)$ e
%	For single user selection with $\|{\bf h}_{k}\| = \sqrt{\gamma},\ \forall k$, selecting a user with $(i)$ the largest number of non-zero channel elements in the beamspace, and  $(ii)$ equal amplitude spread across the elements maximizes the achievable rate.
\end{lemma}
\begin{proof}
%	See Appendix \ref{appx:lemma1}	
The ZF combiner for a single user becomes ${\bf w}_{\rm zf} = {{\bf h}_{\rm b}}/{\|{\bf h}_{\rm b}\|^2}$. 
%	Note that we omit the beamspace channel user index for simplicity and ${\bar {\bf h}}_{\rm b} ={\bf h}_{\rm b}$ for $N_{\rm RF} = N_r$.
Then, \eqref{eq:rate} is given as
\begin{align}
	\nonumber
	\mathcal{R}({\bf h}_{\rm b}) 
	& = \log_2 \left(1+\frac{\alpha \rho}{(1-\alpha)\frac{{\bf h}_{\rm b}^H}{\|{\bf h}_{\rm b}\|^2} {\rm diag}\Big(\rho{\bf h}_{\rm b} {\bf h}_{\rm b}^H + {\bf I}_{N}\Big) \frac{{\bf h}_{\rm b}}{\|{\bf h}_{\rm b}\|^2} + \frac{\alpha}{\|\bh_{\rm b}\|^2}}\right) 
		\\ \label{eq:rate1}
	&= \log_2 \Bigg(1+\frac{\alpha \rho \|{\bf h}_{\rm b}\|^4}{{\rho(1-\alpha)} \sum_{i \in \mathcal{L}}{|h_{{\rm b},i}|^4}+ {\|{\bf h}_{\rm b}\|^2}}\Bigg),
\end{align}
where $\mathcal{L}$ is the set of indices of non-zero complex gains in ${\bf h}_{\rm b}$ with $|\mathcal{L}| = L$.
%	, and ${h_{{\rm b},i}}$ represents the $i$th element of ${\bf h}_{\rm b}$.
%	We need to find the optimal structure of ${\bf h}_{\rm b}$ that maximizes \eqref{eq:rate1} with fixed beamspace channel gain $\|{\bf h}_{\rm b}\|^2 = \gamma > 0$, which leads to the equivalent minimization problem
With the constraint of $\|{\bf h}_{\rm b}\| = \sqrt{\gamma}$, the problem of maximizing $\mathcal{R}({\bf h}_{\rm b})$ in \eqref{eq:rate1} reduces to
\begin{align}
	\label{eq:rate1_pf1}
	\min_{{\bf h}_{\rm b}}  \sum_{i \in \mathcal{L}}{|h_{{\rm b},i}|^4}	\quad \text{s.t. } \|{\bf h}_{\rm b}\|^2 = \gamma.
\end{align}
We use Karush-Kuhn-Tucker conditions to solve the reduced problem in \eqref{eq:rate1_pf1}.
Let $x_i = |h_{{\rm b},i}|^2$ for $i = 1,2,\dots,N$. 
The Lagrangian of the problem with a Lagrangian multiplier $\mu$ is given as
\begin{align}
% 	\label{eq:rate1_pf2}
    \nonumber
	\mathfrak{L}({\bf x}, \mu) = \|{\bf x}\|^2 + \mu \bigg(\sum_{i\in \mathcal{L}} x_i - \gamma \bigg).
\end{align}
By taking a derivative of $\mathfrak{L}({\bf x}, \mu)$ with respect to $x_i$ for $i \in \mathcal{L}$ and setting it to zero, we obtain $x_i = -{\mu}/{2}$.
Putting it to $\sum_{i\in \mathcal{L}} x_i = \gamma $, we have $\mu = -2\gamma/L$.
Finally, the solution becomes
%\vspace{-1em}
\begin{align}
	\label{eq:rate1_pf3}
	x_i = {\gamma}/{L}, \quad i \in \mathcal{L}.
\end{align}
Under the virtual channel representation, $x_i$ indicates the power of the beamspace complex gains and $ L$ is the number of propagation paths.
Accordingly, the physical meaning of \eqref{eq:rate1_pf3} is that the achievable rate for the single user case with the given channel power  $\|{\bf h}_{\rm b}\|^2 = \gamma$ can be maximized when the channel power $\gamma$ is evenly spread to the $L$ beamspace complex gains.
	
By applying the solution $|h_{{\rm b},i}^\star|^2 = \gamma/L$ in  \eqref{eq:rate1_pf3} for $i \in \mathcal{L}$, the achievable rate in \eqref{eq:rate1} becomes 
\begin{align}
	\label{eq:rate1_pf4}
	\mathcal{R}({\bf h}_{\rm b}^\star) 
	= \log_2 \left(1+\frac{\alpha \rho}{{\rho(1-\alpha)}/{L} + {1}/{\gamma}}\right).
\end{align}
The quantization noise variance term in \eqref{eq:rate1_pf4} decreases as $L$ increases.
Therefore, the achievable rate $\mathcal{R}({\bf h}_{\rm b}^\star)$ can be further maximized if the scheduled user channel ${\bf h}_{\rm b}^\star$ has the largest number of propagation paths with equal power distribution across the beamspace complex gains.
%	 when channel gains $\|{\bf h}_{\rm b}\|^2 = \gamma$ are equal for all users.
\end{proof}
	
Unlike the conventional understanding that scheduling a user with the largest channel gain achieves the maximum achievable rate for the single user communication in the noise limited environment, Lemma \ref{lem:rate1} shows that the achievable rate is related not only to the channel magnitude $\|{\bf h}_{\rm b}\|$ but also to the channel structure in the beamspace when received signals are coarsely quantized.
We further show that if the number of propagation paths $L$ is limited, the maximum rate for the single user case converges to a finite value as the channel magnitude increases. 
\begin{corollary}
	\label{cor:finite}
	With the finite number of channel propagation paths $L$, the maximum achievable rate with single user scheduling converges to  
	\vspace{-0.5em}
	\begin{align}
		\label{eq:rate1_inf}
		\mathcal{R}({\bf h}_{\rm b}^\star) 
		\to \log_2 \left(1+{\alpha L}/{(1-\alpha)}\right), \quad \text{as }\|{\bf h}_{{\rm b}}\| \to \infty.
	\end{align}
\end{corollary}
\begin{proof}
	The maximum achievable rate of the single user scheduling with the given $L$ and $\|{\bh_{\rm b}}\|^2 = \gamma$ is derived in \eqref{eq:rate1_pf4}.
	Then, \eqref{eq:rate1_pf4} converges to \eqref{eq:rate1_inf} as increasing the channel gain ($\gamma \to \infty$).
\end{proof}
{\color{black} Corollary \ref{cor:finite} shows that the quantization error ($\alpha < 1$) limits the achievable rate to remain finite because the quantization noise variance also increases with the increase of the channel magnitude. This implies that the conventional scaling law $\log\log K$ \cite{sharif2005capacity} cannot be met in the low-resolution ADCs regime.}
Accordingly, as the SNR increases, mitigation of the quantization error becomes a more critical problem that needs to be considered in user scheduling. 

Now, we investigate the multiuser scheduling for the case where $\sum_{k= 1}^S L_{\mathcal{S}(k)} \leq N$. 
Here, $\cS(k)$ is the $k$th scheduled user.
This condition is relevant to mmWave channels where the number of channel paths $L_k$ is presumably very small \cite{akdeniz2014millimeter}.
%This assumption, indeed, can be naturally aligned with mmWave communications since the number of channel propagation paths is considered to be small compared to the number of antennas $L \ll N_r$.
We solve the problem $\cP 2$ to characterize the channel properties that maximize the scheduling gain.
Theorem \ref{thm:rate_multi} shows the structural scheduling criteria of channels to maximize the sum rate in $\cP 2$ for the considered case.
\begin{theorem}
	\label{thm:rate_multi}
	For the case where $\sum_{k=1}^{S}L_{\mathcal{S}(k)} \leq N$, scheduling a set of users $\mathcal{S}$ that satisfies the following channel characteristics maximizes the uplink sum rate in $\cP2$.
	\begin{enumerate}
		\item[(i)] Unique AoAs at the receiver for the channel propagation paths of each scheduled user:
		\begin{align} 
			\label{eq:unique AoAs}
			\mathcal{L}_{\mathcal{S}(k)} \cap \mathcal{L}_{\mathcal{S}(k')} = \emptyset \text{ if } k \neq k',
		\end{align} 
		where $\cL_{\mathcal{S}(k)}$ represents the set of indices of non-zero complex gains in $\bh_{{\rm b},{\mathcal{S}(k)}}$.
%		$\mathcal{L}_{\mathcal{S}(k)} \cap \mathcal{L}_{\mathcal{S}(k')} = \phi$ if $k \neq k'$.
		\item[(ii)] Equal power spread across the beamspace complex gains within each user channel:
		\begin{align}
			\label{eq:Equal spread}
			|h_{{\rm b},i,\mathcal{S}(k)}| = \sqrt{\gamma_{\mathcal{S}(k)}/L_{\mathcal{S}(k)}} \text{ for } i \in \mathcal{L}_{\mathcal{S}(k)}.
		\end{align} 
	\end{enumerate}
\end{theorem}
\begin{proof}
%	See Appendix \ref{appx:theorem1}
We take a two-stage maximization approach and show the sufficient conditions for maximizing the sum rate in $\cP 2$ with the constraint of $\sum_{k=1}^S L_{\mathcal{S}(k)} \leq N$. 
{\color{black} Using the diagonal structure of $\bR_{\bq\bq}$ as shown in \eqref{eq:cov}, we rewrite \eqref{eq:rate} in a simpler form as}
%We rewrite The achievable rate in \eqref{eq:rate} can be rewritten as
\begin{align}
	\label{eq:rate_multi_pf}
	r_k({\bf H}_{\rm b}) 
	= \log_2 \left(1+\frac{\alpha \rho}{\rho (1-\alpha){\bf w}_{{\rm zf},k}^H {\rm diag}\Big({\bf H}_{\rm b} {\bf H}_{\rm b}^H\Big){\bf w}_{{\rm zf},k} + \|{\bf w}_{{\rm zf},k}\|^2}\right). 
\end{align}
In the first stage, we focus on minimizing $\|{\bf w}_{{\rm zf},k}\|^2$ in \eqref{eq:rate_multi_pf}. 
When user channels are orthogonal, ${\bf h}_{{\rm b},k} \! \perp\! {\bf h}_{{\rm b},k'}$ for $ k \neq k'$, we have ${\bf w}_{{\rm zf},k} = {\bf h}_{{\rm b},k}/\|{\bf h}_{{\rm b},k}\|^2$.
%  ${\bf w}_{{\rm zf},k}^H{\bf h}_{{\rm b},k'}= 0$ for $k \neq k'$.   
Since ${\bf w}_{{\rm zf},k}$ with minimum norm is known as ${\bf w}_{{\rm zf},k} = {\bf h}_{{\rm b},k}/\|{\bf h}_{{\rm b},k}\|^2$, $\|{\bf w}_{{\rm zf},k}\|^2$ can be minimized with the orthogonality condition.
%	 as a ZF combiner.
	
In the second stage, we minimize the achievable rate of \eqref{eq:rate_multi_pf} by imposing the orthogonality condition from the first stage as follows:
\begin{align}
	\label{eq:rate_multi_pf0}
	r_k({\bf H}_{\rm b}|{\bf h}_{{\rm b},k}\! \perp \! {\bf h}_{{\rm b},k'}) 
%		\\ \nonumber
%		&= \log_2 \left(1+\frac{\alpha \rho}{\beta{\bf w}_{{\rm zf},k}^H {\rm diag}\Big(\rho{\bf H}_{\rm b} {\bf H}_{\rm b}^H + \frac{1}{1-\alpha}{\bf I}_{N_{\rm RF}}\Big) {\bf w}_{{\rm zf},k}}\right)  
	& \stackrel{(a)} = \log_2 \Biggl(1+\frac{\alpha \rho {\|{\bf h}_{{\rm b},k}\|^4} }{\rho (1-\alpha){\bf h}^H_{{\rm b},k}{\rm diag}\Big({\bf H}_{\rm b} {\bf H}_{\rm b}^H\Big) {\bf h}_{{\rm b},k}+ \|{\bf h}_{{\rm b},k}\|^2}\Biggr) 
		\\ \nonumber
	& = \log_2 \left (1+\frac{\alpha \rho \gamma^2_k}{\rho (1-\alpha)  \underset {i\in \mathcal{L}_{k}}{\sum} |h_{{\rm b},i,k}|^2\bigg(|h_{{\rm b},i,k}|^2 + \underset{u\neq k}{\sum} |h_{{\rm b},i,u}|^2 \bigg ) + \gamma_k}\right)
		\\	\label{eq:rate_multi_pf1}
	&\stackrel{(b)} \leq \log_2 \left(1+\frac{\alpha \rho \gamma^2_k}{\rho (1-\alpha)  \sum_{i\in \mathcal{L}_{k}}|h_{{\rm b},i,k}|^4 + \gamma_k}\right)
		\\	\label{eq:rate_multi_pf2}
	&\stackrel{(c)} \leq \log_2 \Bigg(1+\frac{\alpha \rho}{ \rho (1-\alpha) /L_k + 1/\gamma_k}\Bigg).
\end{align}
The equality (a) is from ${\bf w}_{{\rm zf},k} = {{\bf h}_{{\rm b},k}}/{\|{\bf h}_{{\rm b},k}\|^2}$.
% and the conditional rate becomes \eqref{eq:rate_multi_pf0}.
The equality in (b) holds if and only if $|h_{{\rm b},i,u}| = 0$, $\forall \, u \neq k$ and $i \in \mathcal{L}_k$.
This implies that each user needs to have channel paths with unique AoAs to maximize the achievable rate.
Note that \eqref{eq:rate_multi_pf1} is equivalent to the achievable rate of the single user scheduling in \eqref{eq:rate1} due to the channel orthogonality and unique AoA conditions.
Consequently, applying Lemma \ref{lem:rate1}, we have the inequality (c) which comes from the fact that \eqref{eq:rate_multi_pf1} is maximized when $|h_{{\rm b},i,k}| = \sqrt{\gamma_k/L_k}$ for $i \in \mathcal{L}_k$, i.e., channel power is spread evenly across the beamspace complex gains within each user channel. 
The upper bound in \eqref{eq:rate_multi_pf2} is equivalent to the maximum achievable rate for the single user case in \eqref{eq:rate1_pf4}. 
Therefore, \eqref{eq:rate_multi_pf2} is also the maximum achievable rate of each user for the problem $\cP 2$, which also maximizes the sum rate in $\cP 2$.
	
Throughout the proof, we have shown that the derived conditions: the orthogonality, the unique AoA, and the equal power spread conditions are sufficient to maximize the sum rate in $\cP 2$ for the case of $\sum_{k=1}^S L_{\mathcal{S}(k)} \leq N$.
Since, the unique AoA condition implies the orthogonality, only the unique AoA and equal power spread conditions are required to be satisfied by the beamspace channel matrix ${\bf H}_{\rm b}$ for maximizing the uplink sum rate.
% with $\sum_{k=1}^S L_{\mathcal{S}(k)} \leq N$.
This completes the proof.
\end{proof}

Distinguished from conventional channels, there are channel orthogonality cases related to mmWave massive MIMO communications: (a) asymptotic orthogonality of array steering vectors across different angles \cite{ngo2014aspects}, (b) 
orthogonality of beamspace channel sub-vectors having common AoAs, and (c) orthogonality of array steering vectors in \eqref{eq:ARV} with  angle offsets of multiples of $2/M$ \cite{lee2016performance}.
Note that the first condition in \eqref{eq:unique AoAs} particularly emphasizes the third case (c) which forces the beamspace channel orthogonality and further minimizes the aggregated channel gain at each RF chain by avoiding overlap between channel gains in the same AoA, which reduces the quantization noise variance as discussed in Remark~\ref{rm:intuition}.
The second condition in \eqref{eq:Equal spread} also minimizes the aggregated channel gain by evenly spreading the channel power across the beamspace gains, and thus, reduces the quantization error.
Consequently, Theorem \ref{thm:rate_multi} emphasizes the importance of the channel structure in maximizing the sum rate under coarse quantization,
% when $L N_s\leq N_r$ and $\|{\bf h}_{{\rm b},k}\| = \sqrt{\gamma} $, $\forall k$.
while conventional user scheduling approaches ignore such criteria.
% as no quantization is considered.
{\color{black} 
Therefore, we propose a quantization-aware scheduling algorithm based on the criteria in Theorem \ref{thm:rate_multi}.
Although the scheduling criteria in Theorem \ref{thm:rate_multi} is derived under the condition of $\sum_{k=1}^{S}L_{\mathcal{S}(k)} \leq N$, we show that the proposed algorithm which exploits the criteria still achieves higher performance compared to conventional algorithms for $\sum_{k=1}^{S}L_{\mathcal{S}(k)} > N$ in Section \ref{sec:simulation}. 
}

% a decent quantization-aware scheduling algorithm needs to be developed for properly extracting scheduling gains in low resolution ADC systems.
%a user scheduling algorithm that is suitable for the systems with low-resolution ADCs needs to be newly developed.
% which degrades the performance of previous scheduling algorithms under coarse quantization.
 % result since the interference term in $\mathcal{R}_k$ is viewed as weighted sum of its channel gain with beamspace gains as weights with an orthogonal set.
%Since channel gains are fixed for each user, we need to minimize beamspace gains to minimize user interference as mentioned in Remark \ref{rm:intuition} when quantization error exists.
%This is different from SUS because SUS only considers orthogonality of user channels, which can be satisfied even with overlapping. 
%Furthermore, it considers the amplitude of the aggregated channel gain $\|{\bf h}_{\rm b}\|$ regardless of how the gains are spread over the beamspace.

%\begin{corollary}
%	For single path channels $L=1$ with equal gains $\|{\bf h}_{{\rm b},k}\| = \sqrt{\gamma}$ for $k = 1,\dots, N_u$, SUS is the optimal user selection.
%\end{corollary}
%\begin{proof}
%	SUS would select non-overlapping users for orthogonality, and the channel gain is considered to be spread evenly since $L =1$.
%	Thus, from Proposition \ref{thm:rate_multi}, it is the optimal user selection.
%\end{proof}

%%%%%%%%%%%%%%%%%%%%%%%%%%%%%%%%%%
\subsection{Proposed Algorithm}
\begin{algorithm}[!t]
\label{algo:CSS}
\vspace{.3em}
% \KwData{this text}
% \KwIn{daf}
% \KwOut{fad}
% \KwResult{how to write algorithm with \LaTeX2e }
 {\bf Initialization}: $\mathcal{K}_1 = \{1,\dots, K\}$, $\mathcal{S} = \phi$, and $i = 1$.\\
 \For{k = 1:K}{
 BS stores $N_b \geq L_k$ indices of dominant spatial angles of ${\bf h}_{{\rm b}, k}$ in $\mathcal{B}_k$ .
 }
{\bf Iteration}: 
 \While{$i \leq S$ and $\cK_i \neq \emptyset$ }{

   \For{$k \in \cK_i$}{
   BS computes approximated SINR of user $k$,  ${\rm SINR}_k\big({\bf H}_{\rm b}(\mathcal{S}\cup \{k\})\big)$ in \eqref{eq:sinr}.
%   \begin{align}
%   		\nonumber
% 		{\rm SINR}_k\big({\bf H}_{\rm b}(\mathcal{S}\cup \{k\})\big) = \frac{\alpha \rho \|{\bf h}_{{\rm b},k}\|^4}{(1-\alpha){\bf h}_{{\rm b},k}^H\, {\bf D}({\bf H}_{\rm b}(\mathcal{S}\cup \{k\}))\, {\bf h}_{{\rm b},k}}
% 	\end{align}
}
    BS schedules user who has the largest SINR as
	\begin{align}
		\label{eq:selection}
		\mathcal{S}(i) = \argmax_{k \in \mathcal{K}_i} {\rm SINR}_k\big({\bf H}_{\rm b}(\mathcal{S}\cup \{k\})\big)
	\end{align}
    and updates scheduled user set $\mathcal{S}= \mathcal{S} \cup \{\mathcal{S}(i) \}$.\\
%   To calculate component of ${\bf h}_{{\rm b}, \mathcal{S}(i)}$ that is orthogonal to subspace ${\rm span}\{{\bf f}_{\mathcal{S}(1)},\dots,{\bf f}_{\mathcal{S}(i-1)}\}$, computes
	Then, BS computes orthogonal component ${\bf f}_{\mathcal{S}(i)}$ for filtering as in \eqref{eq:f}.
% 	\begin{align}
% 		\nonumber
% 		{\bf f}_{\mathcal{S}(i) } =\bigg({\bf I} - \sum_{j = 1}^{i-1} \frac{{\bf f}_{\mathcal{S}(j)}{\bf f}_{\mathcal{S}(j)}^H}{\|{\bf f}_{\mathcal{S}(j)}\|^2} \bigg){\bf h}_{{\rm b},\mathcal{S}(i)}
% 	\end{align}

  Using ${\bf f}_{\mathcal{S}(i)}$ and $\cB_{\mathcal{S}(i)}$, BS filters candidate set $\cK_i$ as in \eqref{eq:semi-orthogonal}
% 	\begin{align}
% 		\label{eq:semi-orthogonal}
% 		\mathcal{K}_{i+1} = \bigg\{k\in \mathcal{K}_i\setminus\{\mathcal{S}(i) \}\ |&\ \frac{|{\bf f}_{\mathcal{S}(i)}^H{\bf h}_{{\rm b},k}|}{\|{\bf f}_{\mathcal{S}(i)}\| \|{\bf h}_{{\rm b},k}\|}< \epsilon_{th}, |\mathcal{B}_{\mathcal{S}(i)} \cap \mathcal{B}_k | \leq N_{\rm OL}  \bigg\}
% 	\end{align}   
  and sets $i = i+ 1$;
 }
 \Return{\ }{Scheduled user set $\cS$}\;
 \caption{Channel Structure-based Scheduling (CSS)}
\end{algorithm}
%%%%%%%%%%%%%%%%%%%%%%%%%%%%%%%%%%%%%%%%%%%%%%%%%%%%%%%%%%%%%%%%%%%%%%%%%%%%%%

In this subsection, we propose a user scheduling algorithm with low complexity by using the criteria in Theorem \ref{thm:rate_multi}.
{\color{black} Adopting a greedy manner, the proposed algorithms make it possible to schedule users without examining all combinations of users.}
% , achieving large performance improvement compared to the conventional approaches and random scheduling in mmWave systems with low-resolution ADCs.
At each iteration, the proposed algorithm schedules a user and reduces the size of a user candidate set $\cK$ through filtering.
To extract user diversity, the algorithm filter the user set $\cK$ by enforcing semi-orthogonality between scheduled user channels, not perfect orthogonality.
In addition to the scheduling criteria in Theorem \ref{thm:rate_multi}, we also apply the orthogonality condition in \eqref{eq:rate_multi_pf0} for the filtering to provide higher precision in the semi-orthogonality. 
%Moreover, it is hard to enforce the orthogonality between channels only with the unique AoA condition in \eqref{eq:unique AoAs} since the phase offset between the AoAs of user channels and the quantized angles of the analog combiner produces non-zero values on most of the elements in $\bh_{\rm b}$.

%For the filtering, we apply both the channel orthogonality condition and the unique AoA condition since 
Algorithm \ref{algo:CSS} describes the proposed scheduling method, called channel structure-based scheduling (CSS).
After each user selection, the proposed algorithm filters the user candidate set $\cK$ by leveraging the orthogonality condition in \eqref{eq:rate_multi_pf0} as in \cite{yoo2006optimality} by utilizing \eqref{eq:f}
\begin{align}
	\label{eq:f}
	{\bf f}_{\mathcal{S}(i) } = {\bf h}_{{\rm b},\mathcal{S}(i) } - \sum_{j = 1}^{i-1}  \frac{{\bf f}_{\mathcal{S}(j)}^H{\bf h}_{{\rm b},\mathcal{S}(i) }}{\|{\bf f}_{\mathcal{S}(j)}\|^2}{\bf f}_{\mathcal{S}(j)}
	=\bigg({\bf I} - \sum_{j = 1}^{i-1} \frac{{\bf f}_{\mathcal{S}(j)}{\bf f}_{\mathcal{S}(j)}^H}{\|{\bf f}_{\mathcal{S}(j)}\|^2} \bigg){\bf h}_{{\rm b},\mathcal{S}(i)}
\end{align}
where ${\bf f}_{\mathcal{S}(i) }$ is the component of ${\bf h}_{{\rm b}, \mathcal{S}(i)}$ that is orthogonal to subspace ${\rm span}\{{\bf f}_{\mathcal{S}(1)},\dots,{\bf f}_{\mathcal{S}(i-1)}\}$.
Unlike the algorithm in \cite{yoo2006optimality} which computes the orthogonal component ${\bf f}_k$  for the entire users in the candidate set, the proposed CSS algorithm calculates ${\bf f}_{\mathcal{S}(i)}$ only for the currently scheduled user $\cS(i)$.
%This makes the beamspace channels of scheduled users closely satisfy the equality in \eqref{eq:rate_multi_pf0}. 
The algorithm also enforces additional spatial orthogonality in the beamspace to the filtered set as in \cite{lee2016performance} by modifying the unique AoA condition in \eqref{eq:unique AoAs}.
Since there can exist phase offsets that lead to more than $L_k$ dominant channel gains in ${\bf h}_{{\rm b},k}$ due to the quantized angles of the analog combiner, the algorithm stores $N_b \geq L_k$ indices of dominant spatial angles in $\cB_k$ and filters the user set $\cK$ by removing users whose angle indices in $\cB_k$ show more than $N_{\rm OL}$ overlaps with those of the scheduled user in $\cB_{\mathcal{S}(i)}$. 
The semi-orthogonality filtering becomes 
\begin{align}
	\label{eq:semi-orthogonal}
	\mathcal{K}_{i+1} = \bigg\{k\in \mathcal{K}_i\setminus\{\mathcal{S}(i) \}\ |&\ \frac{|{\bf f}_{\mathcal{S}(i)}^H{\bf h}_{{\rm b},k}|}{\|{\bf f}_{\mathcal{S}(i)}\| \|{\bf h}_{{\rm b},k}\|}< \epsilon, |\mathcal{B}_{\mathcal{S}(i)} \cap \mathcal{B}_k | \leq N_{\rm OL}  \bigg\}. 
\end{align}

These filtering operations not only reduce the size of the user set $\cK$, but also offer semi-orthogonality between the scheduled users in $\cS$ and the candidate users in $\cK$.
As a result, the filtering leads the ZF combiner to be approximated as ${\bf w}_{{\rm zf},k} \approx {\bf h}_{{\rm b},k}/\|{\bf h}_{{\rm b},k}\|^2$ for a user $k \in \cK$, and 
% Utilizing the approximation ${\bf w}_{{\rm zf},k} \approx {\bf h}_{{\rm b},k}/\|{\bf h}_{{\rm b},k}\|^2$, 
we can approximate the SINR of user $k \in \cK$ with previously scheduled users in $\cS$  as
\begin{align}
	\label{eq:sinr}
	{\rm SINR}_k\big({\bf H}_{\rm b}(\cS\cup\{k\})\big) \approx \frac{\alpha \rho \|{\bf h}_{{\rm b},k}\|^4}{(1-\alpha){\bf h}_{{\rm b},k}^H\, {\bf D}\big({\bf H}_{\rm b}(\cS\cup\{k\})\big)\, {\bf h}_{{\rm b},k}}
\end{align}
where ${\bf D}\big({\bf H}_{\rm b}(\cS\cup\{k\})\big) = {\rm diag}\big(\rho\,{\bf H}_{\rm b}(\cS\cup\{k\}){\bf H}_{\rm b}(\cS\cup\{k\})^H + \frac{1}{1-\alpha}{\bf I}_{N}\big)$.
For a scheduling measure, the proposed algorithm adopts the approximated SINR \eqref{eq:sinr} to incorporate the scheduling criteria in Theorem \ref{thm:rate_multi} with the channel magnitude and orthogonality\footnote{\color{black} By treating the approximate SINR as the true SINR and following the technique used in \cite{yoo2006optimality} and \cite{lee2016performance}, the proposed method can be incorporated with the proportional fairness (PF) policy \cite{Viswanath02IT} for fairness-aware scheduling in a heterogeneous system.}.
At each iteration, the algorithm schedules the user who has the largest SINR among the users in $\cK$ as shown in \eqref{eq:selection}.
%Although this method gives best SINR for the newly scheduled user $\cS(i)$, it may badly degrade the SINR of the previously scheduled users due to the increase in the quantization noise variance $\bD$ coupled with the user channels.
%According to Theorem \ref{thm:rate_multi}, however, the newly scheduled user should avoid concentrating the aggregated channel gain on certain spatial angles and spread the channel power to balance the level of the quantization noise variance across the RF chains.
%Consequently, this can lead to the minimum degradation of the SINR for the previously scheduled users.
Using the approximated SINR \eqref{eq:sinr} for the selection measure greatly reduces the computational complexity by avoiding the matrix inversion for computing the ZF combiner ${\bf W}_{\rm zf}$.
%Since SUS does not take the structure of the channel matrix into account for user selection, it only provides sub-optimal selection under the existence of quantization distortion. 

%% ALGORITHM %%%%%%%%%%%%%%%%%%%%%%%%%%%%%%%%%%%%%%%%%%%%%%%%%%%%%%%%%%%%
%\begin{algorithm}[!t]
%\caption{Greedy User Scheduling}
%\label{algo:greedy}
%\vspace{.3 em}
%\begin{enumerate}%[leftmargin=*]
%	\item BS initializes $\mathcal{T}_1 = \{1,\dots,N_u\}$, $\mathcal{S}_G = \phi$, and $i = 1$.
%%	\vspace{.1 em}
%	\item BS selects a user who maximizes sum rate as
%	\begin{align}
%		%\mathcal{S}(i) = \argmax_{k \in \mathcal{T}}  \mathcal{R}\big([{\bf H}_{\rm b}(\mathcal{S}),{\bf h}_{{\rm b},k}]\big)
%		\mathcal{S}_G(i) = \argmax_{k \in \mathcal{T}_i} \sum_{j \in {\mathcal{S}_G \cup \{k\}}} \mathcal{R}_j\big([{\bf H}_{\rm b}(\mathcal{S}_G),{\bf h}_{{\rm b},k}]\big)
%	\end{align}
%	%with $\mathcal{R}\big([{\bf H}_{\rm b}(\mathcal{S}),{\bf h}_{{\rm b},k}]\big) = \sum_{j \in {\mathcal{S} \cup \{k\}}} \mathcal{R}_j([{\bf H}_{\rm b}(\mathcal{S}),{\bf h}_{{\rm b},k}])$ 
%	where $\mathcal{R}_j$ is given in \eqref{eq:rate}.
%%	\vspace{.5 em}
%	\item Update $\mathcal{T}_{i+1} =\mathcal{T}_i\setminus \{\mathcal{S}_G(i)\}$, $\mathcal{S}_G = \mathcal{S}_G \cup \{\mathcal{S}_G(i) \}$, and $i = i + 1$, and go to step 2 until select $N_s$ users.
%\end{enumerate}
%%\vspace{.5 em}
%\end{algorithm}
%%%%%%%%%%%%%%%%%%%%%%%%%%%%%%%%%%%%%%%%%%%%%%%%%%%%%%%%%%%%%%%%%%%%%%%%%%%%%%%

%%%%%%%%%%%%%%%%%%%%%%%%%%%%%%%%%%%%%%%%%%%%%%%%%%%%%%%%%%%%%%%%%%%%%%%%%%%%%%
\begin{algorithm}[!t]
\label{algo:greedy}
\vspace{.3em}
% \KwData{this text}
% \KwIn{daf}
% \KwOut{fad}
% \KwResult{how to write algorithm with \LaTeX2e }
{\bf Initialization}: $\mathcal{K}_{{\rm G},1} = \{1,\dots,K\}$, $\mathcal{S}_{\rm G} = \emptyset$, and $i = 1$.\\
{\bf Iteration}: 
 \While{$i \leq S_{\rm G}$}{

   \For{$k \in \cK_{{\rm G},i}$}{
   Compute sum rate using $r_j$ in \eqref{eq:rate} for scheduled users and each user $k \in \cK_{{\rm G},i}$ as
   \begin{align}
		\cR_k = \sum_{j \in {\mathcal{S}_{\rm G} \cup \{k\}}} r_j\big({\bf H}_{\rm b}(\mathcal{S}_{\rm G}\cup \{k\})\big)
	\end{align}
%	where $r_j(\cdot)$ is given in \eqref{eq:rate}.
	}
    BS schedules user who maximizes sum rate as $\mathcal{S}_{\rm G}(i) = \argmax_{k \in \mathcal{K}_{{\rm G},i}} \cR_k$
% 	\begin{align}
% 		\mathcal{S}_{\rm G}(i) = \argmax_{k \in \mathcal{K}_{{\rm G},i}} \cR_k
% 	\end{align}
	and 
	
	updates $\mathcal{K}_{{\rm G},i+1} =\mathcal{K}_{{\rm G},i}\setminus \{\mathcal{S}_{\rm G}(i)\}$, $\mathcal{S}_{\rm G} = \mathcal{S}_{\rm G} \cup \{\mathcal{S}_{\rm G}(i) \}$, and $i = i + 1$;
 }
 \Return{\ }{Scheduled user set $\cS_{\rm G}$}\;
 \caption{Greedy Max-Sum Rate Scheduling}
\end{algorithm}
%%%%%%%%%%%%%%%%%%%%%%%%%%%%%%%%%%%%%%%%%%%%%%%%%%%%%%%%%%%%%%%%%%%%%%%%%%%%%%

To provide a reference in sum rate performance, we also propose a high-complexity and high-performance greedy algorithm which schedules the user who achieves the highest sum rate at each iteration as shown in Algorithm \ref{algo:greedy}.
% Although the greedy algorithm offers sub-optimal performance, it requires high computational complexity.
%Algorithm \ref{algo:greedy} describes the greedy method.
At each iteration, the greedy algorithm computes sum rate in \eqref{eq:rate}, i.e., the algorithm computes the exact SINR for scheduled users in $\cS_{\rm G}$ and a candidate user $k$, $\forall k \in \cK_{{\rm G},i}$.
Thus, the algorithm carries the huge burden of computing a matrix inversion $|\cK_{{\rm G},i}|$ times at each selection.
At the $i$th stage, the greedy algorithm computes the achievable rate in \eqref{eq:rate} $|\mathcal{K}_{{\rm G},i}| \times i$ times and compares the derived $|\mathcal{K}_{{\rm G},i}|$ sum rates, whereas the CSS algorithm only computes the approximated SINR in \eqref{eq:sinr} $|\mathcal{K}_i|$ times and compares $|\mathcal{K}_i|$ SINRs.
Moreover, unlike the greedy algorithm, the CSS algorithm reduces the size of the user set $\mathcal{K}_i$ by filtering in \eqref{eq:semi-orthogonal} at each iteration. 
This leads to $|\mathcal{K}_i| \ll |\mathcal{K}_{{\rm G},i}|$, and the gap $|\mathcal{K}_{{\rm G},i}|- |\mathcal{K}_i|$ will increase with iteration; the CSS algorithm becomes more efficient with larger $K$ and\,/or $S$.

{\color{black}
%Although we consider the narrowband channel model, the proposed algorithm can be applied to an OFDM system for a wideband channel case. 
%Given the RF combiner of $\bW_{\rm RF}=\tilde{\bA}$, for each subcarrier $i=1,\cdots,I$, the quantized received signal $\by_q[i]$ is given as
%\begin{align}
%    \nonumber
%    \by_q[i] = \alpha\sqrt{\rho}\bH_{\rm b}[i]\bs[i] + \alpha \bv[i] + \bq[i],
%\end{align}
%where $\bH_{\rm b}[i]=\tilde{\bA}^H\bH[i]$ is the projection of the channel matrix at subcarrier $i$ onto the beamspace. Here, each column $k$ of $\bH[i]$ can be written as \cite{venugopal2017channel}
%\begin{align}  
%    \nonumber
%    \bh_k[i] = \sqrt{\frac{M}{L_k}}\sum_{\ell=1}^{L_k}g_{k,\ell}\omega_{\tau_{k,\ell}}[i]\ba(\phi_{k,\ell}),
%\end{align}
%where $\omega_{\tau_{k,\ell}}[i]=\sum_{r=0}^{R-1}p(rT_s-\tau_{k,\ell})e^{-{\rm j}\frac{2\pi i r}{I}}$, $R$ is the number of delay taps, and $p(\tau)<\infty$ denotes the combined effects of pulse shaping and analog filtering for $T_s$-spaced sampling at $\tau$.
\begin{remark}
    The proposed algorithm can be applied to an orthogonal frequency division multiplexing (OFDM) system for a wideband channel case. Since we consider the system with a given analog combiner, the proposed algorithm can be performed independently for each subcarrier index $i$. However, the structure of the quantization noise $\bq[i]$ in the wideband OFDM system becomes different from that of the narrowband system so that the spatial filtering in the proposed user scheduling algorithm may not be desirable. 
    Nonetheless, the approximated SINR can still be applicable with the semi-orthogonality filtering by computing the quantization noise variance for each subcarrier $i$ of the OFDM system $\bR_{\bq\bq}[i]$.
    Thus, the BS can perform the proposed algorithms to schedule users to be served on each subcarrier by relaxing the spatial filtering.
\end{remark}
}

We note that the proposed method schedules users with minimum overlap among quantized AoAs of user channels to satisfy the derived scheduling criterion (i) in Theorem \ref{thm:rate_multi}. 
Accordingly, by using the proposed scheduling method, the beamforming-based Doppler effect reduction techniques such as a per-beam synchronization approach in \cite{you2017bdma} can be performed at the BS since the BS can see each beam with a single dedicated user signal with large channel gains and possibly with other user signals with negligible channel gains. 
Therefore, the proposed scheduling method can provide potential benefit in reducing Doppler effect when jointly used with Doppler effect mitigation techniques at the BS.

{\color{black}
\subsection{Beam Training-Based Channel Acquisition}\label{subsec:CQI}
Assuming time-division duplex communications, we briefly provide an example of extension our algorithm to a practical system where the BS uses beam training and receives channel quality indicators (CQIs) from users. 
A procedure of beam training and CQI feedback can be as follows:
\begin{enumerate}
    \item \noindent The BS constructs a set of $N_s \geq N$ beam vectors $\{\ba(\bar\vartheta_1),\dots,\ba(\bar\vartheta_{N_s})\}$ with the angles within the angles of the analog combiner $\tilde\bA$, i.e., there exists $i$ such that $\bar\vartheta_n \in [\vartheta_i - 1/M, \vartheta_i + 1/M]$, $\forall n$, where $\vartheta_i$ is the spatial angle of the $i$th analog beamformer.
    Then, the BS transmits each beam of the set in time to all users in the cell during a training phase. 
    \item Each user $k$ can estimate the channel gain corresponding to each beam and have the estimate of $\bh_k^H\tilde{\bA}=\bh_{{\rm b},k}^H$ at the end of the beam training. 
    From the sparsity of the mmWave channel, few elements of $\bh_{{\rm b},k}$ have non-negligible beam gains and we can implement an efficient feedback method that exploits the sparsity of the effective channel $\bh_{{\rm b},k}$ as described in \cite{he2017codebook}.
    For instance, each user can feed back the beam indices of the non-negligible beam gains and their corresponding channel coefficients in a quantized form to the BS.
    % Depending on the capacity of amount of feedback, each user may feed back only the beam indices, which requires only few integer numbers. 
    \item After the feedback from all users is over, the BS can create an estimate of $\bH_b$ with the feedback information by simply padding zeros in the unreported beam indices. Then, the BS can directly apply the proposed scheduling algorithm by using the estimated channel.  
\end{enumerate}

The estimation error with the CQI feedback is expected to degrade both the proposed algorithm and conventional scheduling algorithms since full CSI is required for all approaches. 
We leave the analysis of the imperfect CSI case for a future work as it is beyond the scope of our work. 
}

%%%%%%%%%%%%%%%%%%%%%%%%%%%%%%%%%%%%%%%%%%%%%%%%%%%%%%%%%%%%%%%%%%%%
\section{User Scheduling with Partial Channel Information}
\label{sec:pCSIR}
%%%%%%%%%%%%%%%%%%%%%%%%%%%%%%%%%%%%%%%%%%%%%%%%%%%%%%%%%%%%%%%%%%%%

In this section, we propose a user scheduling algorithm when only partial CSI is known at the BS since it can be challenging to obtain reliable instantaneous CSI estimates for entire users as the number of antennas or users becomes large.
A reasonable alternative is to use slowly-varying channel characteristics, in particular, AoAs of mmWave channels \cite{park2017spatial};
%Since AoAs persist over longer than the coherence time of mmWave channels and mmWave channels have a limited number of propagation paths, it can be easier to estimate AoAs and we can reduce channel estimation overhead for user scheduling by exploiting AoAs compared to instantaneous CSI. 
AoAs persist over longer than the coherence time of mmWave channels, and mmWave channels have a limited number of AoAs.
In this regard, by using the AoA knowledge, the proposed algorithm can greatly reduce the burden of estimating instantaneous full CSI at each channel coherence time.
%The AoA is also related to the longterm channel statistics, particularly spatial channel covariance as\cite{love2008overview, park2017exploiting} estimated at the BS in mmWave communication channels \cite{alkhateeb2014channel, lee2014exploiting}.
%A reasonable alternative to instantaneous CSIT is to use longterm channel statistics, in particular the spatial channel covariance, to configure the analog precoder. Firstly, spatial channel covariances vary over a longer time scale compared to the instantaneous channels, which makes it easier to estimate.
%This assumption would be more practical 
After scheduling, we assume that the BS acquires the effective CSI of the scheduled users for decoding.
% We further analyze the performance of the algorithm by deriving ergodic rates. 
%under the proposed user scheduling method.
%In such case, a question naturally arises, asking which approach gives more benefit for user scheduling between imperfect channel estimation of a small number of users and perfect channel estimation of a large number of user.
%The latter would be a more practical solution in mmWave communication channels since the slowly-changing channel characteristics such as AoAs can be estimated at the BS.
%The derived ergodic rates offer an analytical performance guideline for the proposed scheduling algorithm in the considered mmWave communication environment.

%Since the BS only has the knowledge of AoAs of users, the scheduling algorithm needs to properly exploit the AoAs. 

%%%%%%%%%%%%%%%%%%%%%%%%%%%%%%%%%%
\subsection{Proposed Algorithm}
%%%%%%%%%%%%%%%%%%%%%%%%%%%%%%%%%%

According to \eqref{eq:channel_geo}, the channel $\bh_k$ lies in the subspace spanned by its array response vectors, i.e., $\bh_k \in {\rm span}\{\ba(\phi_{k,1}),\dots,\ba(\phi_{k,L_k})\}$.
To measure the separation between the subspaces, we adopt chordal distance which measures the angle between the subspaces.
In the initialization phase, the algorithm removes users whose AoAs are not in the range of angles of RF chains (reduced range of angles)\footnote{The range of angles of RF chains indicates the set of angles corresponding to $\bigcup_{i}\{\vartheta:|\vartheta -\vartheta_i|<\frac{1}{M}\}$, i.e., the AoAs in the reduced range of angles are $\varphi_{k,\ell} \in  \bigcup_{i}\{\vartheta:|\vartheta -\vartheta_i|<\frac{1}{M}\}$.}
%\footnote{The range of angles of RF chains indicates the set of angles corresponding to the main lobes of the angles $\{\vartheta_i\}$ of the analog combiner ${\bf W}_{\rm RF}=\tilde{\bf A}$, i.e., the AoAs in the reduced range of angles are $\vartheta_{k,\ell} \in  \bigcup_{i}\{\vartheta:|\vartheta -\vartheta_i|<\frac{1}{M}\}$, where $\vartheta_{k,\ell} = \frac{2d}{\lambda} \sin(\phi_{k,\ell})$.}
%\footnote{The range of angles of RF chains indicates the set of  angles corresponding to the main lobes of the angles $\{\vartheta_i\}$ of the analog combiner ${\bf W}_{\rm RF}=\tilde{\bf A}$. The main lobe of an angle $\vartheta$ is usually defined as $\{\theta:|\theta-\vartheta|\le \frac{1}{M}\}$.} 
from the initial candidate user set $\cK_{{\rm cd},1}$.
%initializes the candidate user set $\cK_{{\rm cd},1} = \{1,\dots, K\}$, scheduled user set $\cS_{\rm cd} = \phi$ and scheduling count $i = 1$.
In the scheduling phase, a first user is scheduled by randomly selecting a user among the set of users with the most AoAs in the reduced range of angles.
%Considering the array response vectors that correspond to the AoAs as the basis of each user channel, 
To schedule a next user, the algorithm updates the candidate user set $\cK_{{\rm cd},i}$ by filtering users whose chordal distance is shorter than the threshold $d_{\rm th}$ to impose semi-orthogonality among scheduled users. 
Due to the filtering, the remaining users in $\cK_{{\rm cd},i+1}$ are guaranteed to have a certain level of orthogonality with the scheduled users $\cS(j)$ for $j = 1,2,\dots, i-1$.
Then, the algorithm schedules the user with the longest chordal distance among the remaining users with the most AoAs in the reduced range of angles.
% \begin{align}
%   	\label{eq:max_dcd}
%   	\cS_{\rm cd}(i) = \argmax_{k \in \mathcal{U}}\, d_{\rm cd}\left({\mathcal{S}_{\rm cd}(i-1)},k\right)
% \end{align}
% where $\cU$ is the set of users with most AoAs in the filtered set $\cK_{{\rm cd},i+1}$.

To this end, we generate the matrix of array response vectors for each user by exploiting the AoA knowledge as ${\bf A}_k = [{\bf a}(\phi_{k,\mathcal{V}_k(1)}),\dots,{\bf a}(\phi_{k,\mathcal{V}_k(V_k)})]$ where $\cV_k$ is the set of AoAs indices within the reduced range of angles for user $k$ and $V_k = |\cV_k|$.
%Let $\cV_k$ be set of AoAs indices that fall into range of angles of RF chains for user $k$.
% where $\phi_{k,\ell}$ is the AoA of each channel propagation path for user $k$.
Let $\cA_k = {\rm span}\{\bA_k\}$ is the subspace for user $k$.
The chordal distance between the two subspaces $(\cA_k$, $\cA_{k'})$ is defined as $d_{\rm cd}(k, {k'}) = \sqrt{\sum_{\ell=1}^{L_{min}}\sin^2\theta_\ell}$
%two subspaces, ${\rm span}\{\bA_k\}$ and ${\rm span}\{\bA_{k'}\}$ where $k\neq k'$, is defined as
%\begin{align}
%	\label{eq:chordal_definition}
%	d_{\rm cd}(k, {k'}) = \sqrt{\sum_{\ell=1}^{L_{min}}\sin^2\theta_\ell}.
%\end{align}
where $L_{min} = \min\{L_k,L_{k'}\}$ and $\theta_\ell \leq \pi/2$ is the principal angle between $\cA_{k}$ and $\cA_{k'}$.
Let ${\bf Q}_k$ be the unitary matrix whose columns are orthonormal basis vectors of $\cA_k$.
%which can be obtained through Gram-Schmidts orthogonalization procedure for $\bA_k$.
{\color{black} According to  \cite{conway1996packing}, we rewrite $d_{\rm cd}(k,k')$ as 	$d_{\rm cd}\left({k},{k'}\right) = \sqrt{L_{min}-{\rm tr}\left({\bf Q}^H_{k}{\bf Q}_{k'}{\bf Q}^H_{k'}{\bf Q}_{k}\right)}$. }
% \begin{align}
% 	\nonumber %\label{eq:chordal_rewrite}
% 	d_{\rm cd}\left({k},{k'}\right) = \sqrt{L_{min}-{\rm tr}\left({\bf Q}^H_{k}{\bf Q}_{k'}{\bf Q}^H_{k'}{\bf Q}_{k}\right)}.
% \end{align}
%Using the chordal distance in \eqref{eq:chordal_rewrite}, we schedule a user who has the longest chordal distance from a previously scheduled user.
%Then, the algorithm updates the candidate user set $\cK_{{\rm cd},i}$ by filtering users whose chordal distance is shorter than the threshold $d_{\rm th}$ to impose semi-orthogonality among scheduled users.
%In the proposed algorithm, the chordal distance measures the angle between two matrices 
The proposed chordal distance-based user scheduling method is described in Algorithm \ref{algo:chordal}.

%%%%%%%%%%%%%%%%%%%%%%%%%%%%%%%%%%%%%%%%%%%%%%%%%%%%%%%%%%%%%%%%%%%%%%%%%%%%%%
\begin{algorithm}[!t]
\label{algo:chordal}
\vspace{.3em}
% \KwData{this text}
% \KwIn{daf}
% \KwOut{fad}
% \KwResult{how to write algorithm with \LaTeX2e }
 {\bf Initialization}: $\mathcal{K}_{{\rm cd},1} = \{1,\dots, K\}$, $\mathcal{S}_{\rm cd} = \phi$, and $i = 1$\\
 \For{k = 1:K}{
 Let $\cV_k$ be set of AoA indices in range of angles of steering vectors for user $k$. \\
 If  $\cV_k= \emptyset$, do $\cK_{{\rm cd},1} =\cK_{{\rm cd},1} \setminus \{k\} $, otherwise, set ${\bf A}_k = [{\bf a}(\phi_{k,\mathcal{V}_k(1)}),\dots,{\bf a}(\phi_{k,\mathcal{V}_k(V_k)})]$.\\
 Generate unitary matrix ${\bf Q}_k = $ column basis of ${\bf A}_k$.}
 {\bf Iteration}: 
 \While{$i \leq S_
 {\rm cd}$ and $\cK_{{\rm cd},i} \neq \emptyset$ }{
  \eIf{i = 1}{
   Randomly schedule first user $\cS_{\rm cd}(1) \in \cK_{{\rm cd},1}$ among users with largest $|\cV_k|$.\\
   Update candidate user set $\cK_{{\rm cd},2} = \cK_{{\rm cd},1} \setminus \cS_{\rm cd}(1)$ and $\cS_{\rm cd} = \cS_{\rm cd} \cup \{\cS_{\rm cd}(1)\}$.  \\
   }{
   \For{$k \in \cK_{{\rm cd},i}$}{
   Let $L_{min} = \min\{L_{\mathcal{S}_{\rm cd}(i-1)},L_k\}$, and compute chordal distance as
%   between previously scheduled user and user $k \in \cK_{{\rm cd},i}$
   \begin{align}
   		\label{eq:dcd_algorithm}
   		d_{\rm cd}\left({\mathcal{S}_{\rm cd}(i-1)},k\right) = \sqrt{L_{min}-{\rm tr}\left({\bf Q}^H_{\mathcal{S}_{\rm cd}(i-1)}{\bf Q}_{k}{\bf Q}^H_{k}{\bf Q}_{\mathcal{S}_{\rm cd}(i-1)}\right)}.
   \end{align}
}
Filter candidate user set based on chordal distance computed in \eqref{eq:dcd_algorithm}
   \begin{align}
   		\label{eq:candidate_set_update}
 	  	\cK_{{\rm cd},i+1} = \big\{k\in \cK_{{\rm cd},i}\ \big| \ {{d_{\rm cd}\left({\mathcal{S}_{\rm cd}(i-1)},k\right)}}/{\sqrt{L_{min}}}> d_{\rm th} \big\}.
   \end{align}\\
   Let $\cU$ be set of users with largest $|\cV_k|$, $\forall k \in \cK_{{\rm cd},i+1}$.
   Schedule user in $\cU$ as
   \begin{align}
   		\label{eq:max_dcd}
   		\cS_{\rm cd}(i) = \argmax_{k \in \mathcal{U}}\, d_{\rm cd}\left({\mathcal{S}_{\rm cd}(i-1)},k\right).
   \end{align}\\
%   Update candidate user set $\cK_{{\rm cd},i} = \cK_{{\rm cd},i}\setminus\cS_{\rm cd}(i)$. 
	Then, update $\mathcal{K}_{{\rm cd},i+1} = \mathcal{K}_{{\rm cd},i+1}\setminus \{\cS_{\rm cd}(i)\}$ and $\cS_{\rm cd} = \cS_{\rm cd} \cup \{\cS_{\rm cd}(i)\}$. 
	   }
    Set $i = i+ 1$;
 }
 \Return{\ }{Scheduled user set $\cS_{\rm cd}$}\;
 \caption{Chordal Distance-based User Scheduling}
\end{algorithm}
%%%%%%%%%%%%%%%%%%%%%%%%%%%%%%%%%%%%%%%%%%%%%%%%%%%%%%%%%%%%%%%%%%%%%%%%%%%%%%

Let $\tilde{\bf h}_k = \sqrt{\frac{M}{L_k}}\sum_{i \in \mathcal{V}_k} g_{k,i}{\bf a}(\phi_{k,i})$.
% be the channel vector produced by using the AoAs within the reduced range of angles for user $k$. 
Then, the algorithm provides an opportunity to schedule users with nearly $\tilde{\bf h}_k \perp \tilde {\bf h}_{k'}$ while the effective channel that the BS sees is the beamspace channel ${\bf h}_{{\rm b},k} = {\bf W}_{\rm RF}^H{\bf h}_k$.
% which is the product of ${\bf W}_{\rm RF}$ and the actual channel ${\bf h}_k$. 
%Here, $\cV_k$ is the set of AoAs that fall into the reduced range of angles for user $k$.
Since the AoAs $\phi_{k,i}$, $i \in \cV_k$ are in the range of angles of RF chains, $\tilde {\bf h}_{k}$ can be regarded to be in the subspace of ${\bf W}_{\rm RF}$, i.e., almost $\tilde {\bf h}_{k} \in {\rm span}\{{\bf W}_{\rm RF}\}$\footnote{If the AoAs  of $\tilde{\bf h}_k$ exactly align with the quantized angles of the analog combiner, $\tilde{\bf h}_k$ perfectly lies in the subspace of ${\bf W}_{\rm RF}$.}.
{\color{black} Accordingly, using ${\bf W}_{\rm RF}^H{\bf W}_{\rm RF} = {\bf I}_N$ which comes from the definition i.e., a sub-matrix of the DFT matrix $\bW_{\rm RF} = \tilde \bA$, we can rewrite $\tilde {\bf h}_{k}$ as }
\begin{align}\label{eq:h_k_projection}
\tilde{\bf h}_k \approx {\bf W}_{\rm RF}({\bf W}_{\rm RF}^H{\bf W}_{\rm RF})^{-1}{\bf W}_{\rm RF}^H\tilde{\bf h}_k = {\bf W}_{\rm RF}{\bf W}_{\rm RF}^H \tilde{\bf h}_k
\end{align}
In addition, we have ${\bf h}_{{\rm b},k} = {\bf W}_{\rm RF}^H{\bf h}_k \approx {\bf W}_{\rm RF}^H\tilde{\bf h}_k$ as the impact of ${\bf a}(\phi_{k,j})$, $\forall j \notin \cV_k$ on the beam domain channel ${\bf h}_{b,k}$ is relatively small compared to that of ${\bf a}(\phi_{k,i})$, $\forall i \in \cV_k$ after analog combining.
In this regard, as the algorithm gives $\tilde{\bf h}_k \perp \tilde {\bf h}_{k'}$, we can nearly have ${\bf h}_{{\rm b},k}  \perp {\bf h}_{{\rm b},k'}$ by 
\begin{align*}
	\epsilon = \tilde{\bf h}_{k}^H\tilde{\bf h}_{k'}	
	\overset{(a)}{\approx} \tilde{\bf h}_{k}^H{\bf W}_{\rm RF}{\bf W}_{\rm RF}^H{\bf W}_{\rm RF}{\bf W}_{\rm RF}^H\tilde{\bf h}_{k'} 
	= \tilde{\bf h}_{k}^H{\bf W}_{\rm RF}{\bf W}_{\rm RF}^H\tilde{\bf h}_{k'} 
	\overset{(b)}{\approx} {\bf h}_{{\rm b},k}^H{\bf h}_{{\rm b},k'}
\end{align*}
where $(a)$ is from \eqref{eq:h_k_projection} and $(b)$ is from ${\bf h}_{{\rm b},k} \approx {\bf W}_{\rm RF}^H\tilde{\bf h}_k$. 
Thus, the proposed algorithm guarantees a certain level of orthogonality between the beamspace channels of the scheduled users.

{\color{black} As discussed in Section \ref{subsec:CQI}, the beam indices for non-negligible channel gains can be obtained by using CQI feedback, i.e., AoAs can be estimated for each user. 
When the capacity of amount of feedback is limited and small, such beam index-only feedback which requires only few integer numbers can be applied to faciliate the proposed chordal distance-based algorithm.}

%Thus, the proposed algorithm guarantees that the beamspace channels of the schedules users have a certain level of orthogonality to each other.
%According to the proposed algorithm, it can be shown that the channels retains a similar level of orthogonality both in the time and beam domain; the algorithm 
%Note that the chordal distance can be applied to the extended case where $\cA_k$ and $\cA_{k'}$ have different dimensions \cite{ko2012multiuser}, i.e., $\bh_k$ and $\bh_{k'}$ can have the different number of channel propagation paths.  

%{\bf Unlike the conventional belief that it is better to have perfect channel information for small number of users than to have partial channel information for large number of users, we numerically show that the opposite can be true under coarse quantization.
%Accordingly, the derived ergodic rates provide the upper bound for both channel state information assumptions.}

%%%%%%%%%%%%%%%%%%%%%%%%%%%%%%%%%%
\subsection{Ergodic Rate Analysis}
%%%%%%%%%%%%%%%%%%%%%%%%%%%%%%%%%%

We now analyze the performance of the chordal distance-based algorithm in ergodic rate.
We focus on the case where each channel has a single propagation path, which corresponds to the sparse nature of mmWave channels \cite{lee2016randomly}, and the number of RF chains are equal to the number of antennas $N = M$ in the analysis.
% For $L = 1$, the proposed algorithm can be simplified as shown in the following remark:
\begin{remark}
	\label{rm:chordal_L1}
	When there is a single path for each user channel, the filtering in \eqref{eq:candidate_set_update} reduces to $\cK_{{\rm cd},i+1} = \big\{k\in \cK_{{\rm cd}, i}  \ \big|\ |{\bf a}^H(\phi_{\mathcal{S}(i-1)}){\bf a}(\phi_k)| < \epsilon_{\rm th} \big\}$ where $\epsilon_{\rm th} \ll 1$, and the scheduling problem in \eqref{eq:max_dcd} becomes $\cS_{\rm cd}(i) = \argmin_{k \in \mathcal{K}_{{\rm cd},i+1}} |{\bf a}^H(\phi_{\mathcal{S}(i-1)}){\bf a}(\phi_k)|$.
\end{remark}
Based on Remark \ref{rm:chordal_L1}, we derive closed-form expressions of the ergodic sum rate for two different cases: (1) AoAs of channels exactly align with the quantized angles of the analog combiner, and (2) channels have arbitrary AoAs regardless of the quantized angles of the analog combiner. 
For the first case, there is no channel leakage in the beamspace and thus, it is often considered as a more favorable channel condition since it improves communication performance such as channel estimation accuracy \cite{sung2018narrowband} and achievable rate \cite{alkhateeb2014mimo, el2014spatially}.
% techniques that exploit the sparsity of channels \cite{sung2018narrowband} can be more accurate, and the achievable rate in hybrid MIMO systems \cite{alkhateeb2014mimo, el2014spatially} can increase as the analog combiner captures most of the channel gains. 
% lee2014exploiting
% The derived ergodic rate for the first case is shown in Proposition \ref{prp:ergodicL1}.

\begin{proposition}
\label{prp:ergodicL1}
	When AoAs of channels exactly align with the quantized angles of the analog combiner with a single propagation path, the ergodic sum rate for $|\cS_{\rm cd}| = S$ scheduled users with the proposed chordal distance-based scheduling algorithm is derived as
	\begin{align}
		\label{eq:ergodicL1}
		%\mathcal{R}_{(L=1)} \leq \log_2 \Bigg(1+ \frac{\alpha}{\rho (1-\alpha)^2}\bigg( \rho (1-\alpha) - e^{\frac{1}{\rho(1-\alpha)}}\Gamma \Big(0,\frac{1}{\rho(1-\alpha)}\bigg)\bigg)\Bigg)
%		\bar{\mathcal{R}}_1 = \bbE\Big[\cR_1\big(\bH_{\rm b}(\cS_{\rm cd})\big)\Big]\approx \frac{S}{\ln 2}\left(e^\frac{1}{\rho M}\,\Gamma\left(0,\frac{1}{\rho M}\right)-e^\frac{1}{\rho (1-\alpha) M }\,\Gamma\left(0,\frac{1}{\rho (1-\alpha)M}\right)\right)
		\bar{\mathcal{R}}_1 =\frac{S}{\ln 2}\left(e^\frac{1}{\rho M}\,\Gamma\left(0,\frac{1}{\rho M}\right)-e^\frac{1}{\rho (1-\alpha) M }\,\Gamma\left(0,\frac{1}{\rho (1-\alpha)M}\right)\right)
	\end{align}
	where $\Gamma(a,z)$ is an incomplete gamma function defined as $\Gamma(a,z) = \int_{z}^{\infty} t^{a-1}e^{-t}\,dt$.
\end{proposition}
%We note that the derived ergodic rate \eqref{eq:ergodicL1} is the sum of the functions without quantization gain and with quantization gain.

%\begin{proposition}
%\label{prp:ergodicL1}
%	When only channel sparsity of all users is known at the receiver, the ergodic sum achievable rate for $N_s$ selected users with single-path channels $(L = 1)$ and quantization distortion $(\alpha < 1)$ is upper bounded by 
%	\begin{align}
%	\label{eq:ergodicL1}
%		%\mathcal{R}_{(L=1)} \leq \log_2 \Bigg(1+ \frac{\alpha}{\rho (1-\alpha)^2}\bigg( \rho (1-\alpha) - e^{\frac{1}{\rho(1-\alpha)}}\Gamma \Big(0,\frac{1}{\rho(1-\alpha)}\bigg)\bigg)\Bigg)
%		\mathcal{R}^{s} \leq N_s \log_2 \Bigg(\frac{1}{1-\alpha} - \frac{\alpha \, e^{\frac{1}{M\rho(1-\alpha)}}}{M\rho(1-\alpha)^2}\Gamma \Big(0,\frac{1}{M\rho(1-\alpha)}\bigg)\Bigg),\quad \alpha < 1
%	\end{align}
%	where $\Gamma(a,z)$ is an incomplete gamma function defined as $\Gamma(a,z) = \int_{z}^{\infty} t^{a-1}e^{-t}\,dt$.
%\end{proposition}
\begin{proof}
	See Appendix \ref{appx:proposition1}.
\end{proof}
\begin{corollary}
	The derived ergodic rate \eqref{eq:ergodicL1} can be expressed as the sum of the ergodic rate without quantization error $\bar{\cR}_{\rm inf}$ and the ergodic rate loss due to quantization error $\bar{\cR}_{\rm loss}(\alpha)$
	\begin{align}
		\bar\cR_1 = \bar{\cR}_{\rm inf} + \bar{\cR}_{\rm loss}(\alpha)
	\end{align}
	 where $\bar{\cR}_{\rm inf} = \frac{S}{\ln 2}e^{\frac{1}{\rho M}} \Gamma (0,\frac{1}{\rho M})$ and $\bar{\cR}_{\rm loss}(\alpha) = -\frac{S}{\ln 2}e^{\frac{1}{\rho(1-\alpha)M}} \Gamma (0,\frac{1}{\rho(1-\alpha) M})$.
\end{corollary}
\begin{proof}
	We can remove the quantization error term in \eqref{eq:ergodicL1_pf} by having $\alpha = 1$. 
	Then, we have 
	\begin{align}
		\bbE\left[\log_2\left(1+\rho \|\bh_{{\rm b},k}\|^2\right)\right] = \frac{1}{\ln 2}e^{\frac{1}{\rho M}} \Gamma \left(0,\frac{1}{\rho M}\right)
	\end{align}
	as $\frac{1}{M}\|\bh_{{\rm b},k}\|^2 = |g_k|^2 \sim {\rm Exp}(1)$, and the ergodic sum rate becomes $\bar{\cR}_{\rm inf} = \frac{S}{\ln 2}e^{\frac{1}{\rho M}} \Gamma (0,\frac{1}{\rho M})$.
\end{proof}
Note that as the number of quantization bits decreases to zero, $\bar{\cR}_{\rm loss}(\alpha)$ increases to $\bar{\cR}_{\rm inf}$, which leads $\bar{\cR}_1 \to 0$.
On the other hand, as the number of quantization bits increases to infinity, $\bar{\cR}_{\rm loss}(\alpha)$ decreases to zero, which leads $\bar{\cR}_1 \to \bar{\cR}_{\rm inf}$. 
This complies with intuition.

Now, we focus on the second case where channels have arbitrary AoAs, which leads to the channel leakage effect in the beam domain due to phase offsets.
% The channel leakage is often considered to be unfavorable in communication systems as it deteriorates channel estimation accuracy and prevents hybrid MIMO systems from conserving full channel gains in the small number of RF chains.
% In Section \ref{sec:sim_partial}, however, we show that such channel leakage effect can achieve the higher ergodic rate than the channel without leakage under the coarse quantization system.
%In practice, however, there exists leakage in beamspace projection of the channel matrix due to the quantized beamforming angles. 
%In this respect, we analyze the ergodic rate with the leakage effect in the large antenna array regime when using the proposed scheduling algorithm with AoAs only.
The derived ergodic rate for the second case is shown in Proposition \ref{prp:leakage_single}.
\begin{proposition}
	\label{prp:leakage_single}
	When channels have a single path and arbitrary AoAs regardless of the quantized angles of the analog combiner, a lower bound of the ergodic sum rate for $|\cS_{\rm cd}| = S$ scheduled users with the proposed chordal distance-based scheduling algorithm is approximated as
	\begin{align}
		\nonumber
		\bar{\cR}_2^{lb} =  \frac{S}{\ln 2}\Biggl(&e^\frac{1+\rho(1-\alpha)(S-1)M^2\mathcal{F}_2(M)}{\rho\alpha M + \rho (1-\alpha)M^2 \mathcal{F}_1(M)}\Gamma\left(0,\frac{1+\rho(1-\alpha)(S-1)M^2\mathcal{F}_2(M)}{\rho\alpha M + \rho (1-\alpha)M^2 \mathcal{F}_1(M)}\right)
		\\ \label{eq:leakage_single}
		 &- e^\frac{1+\rho(1-\alpha)(S-1)M^2\mathcal{F}_2(M)}{\rho(1-\alpha)M^2\mathcal{F}_1(M)}\Gamma\left(0,\frac{1+\rho(1-\alpha)(S-1)M^2\mathcal{F}_2(M)}{\rho(1-\alpha)M^2\mathcal{F}_1(M)}\right)\Biggr)
	\end{align} 
%	\begin{align}
%		\tilde{ \mathcal{R}} = N_s\log\left(1+ \frac{N_r\rho\alpha}{N_r^2\rho(1-\alpha)\left(\mathcal{F}_1(N_r)+ (N_s -1)\mathcal{F}^2_2(N_r)\right)+1}\right)
%	\end{align}
	where $\mathcal{F}_1(M) = \int_0^{1} F^4(\delta, M) \,d\delta$, $\mathcal{F}_2(M) = \left(\int_0^{1} F^2(\delta, M) \,d\delta\right)^2$, and $F(\delta,M)$ is the Fej\'{e}r kernel.
\end{proposition}

\begin{proof}
	See Appendix \ref{appx:proposition2}.
\end{proof}

\begin{remark}
	\label{rm:bit_increase}
	The derived ergodic rate expressions in \eqref{eq:ergodicL1} and \eqref{eq:leakage_single} both converge to $\bar \cR_{\rm inf}$ as the number of quantization bits increases:
	\begin{align}
	    \nonumber
		\bar{\cR}_1, \ \bar{\cR}_2^{lb} \to \frac{S}{\ln 2}e^\frac{1}{\rho M}\,\Gamma\left(0,\frac{1}{\rho M}\right), \quad \text{as }\alpha \to 1.  
	\end{align}
\end{remark}
As the quantization precision increases, the lower bound in \eqref{eq:leakage_proof2} becomes an exact expression, and \eqref{eq:leakage_single} becomes an approximation of the ergodic rate itself rather than its lower bound.
Accordingly, it can be inferred from Remark \ref{rm:bit_increase}  that the two channel scenarios lead to different ergodic rates as a consequence of quantization. 
% Based on Remark \ref{rm:bit_increase}, in Section \ref{sec:sim_partial}, we show that such channel leakage effect can achieve the higher ergodic rate than the channel without leakage under the coarse quantization system.
% while the channel leakage is often considered to be unfavorable in communication systems.
% as it deteriorates channel estimation accuracy and prevents hybrid MIMO systems from conserving full channel gains in the small number of RF chains.
% .
{\color{black}In this regard, although a single path channel is considered, Propositions \ref{prp:ergodicL1} and \ref{prp:leakage_single} still convey meaningful information as they not only provide closed-form ergodic rates but also specify the channel leakage effect in terms of ergodic rate for low-resolution ADCs.
In addition, the single-path channel model is relevant to the case of unmanned aerial vehicle systems \cite{Zeng18arXiv}, which is of interest in upcoming 5G wireless communication systems.
In Section \ref{sec:simulation}, based on the intuition from Propositions \ref{prp:ergodicL1} and \ref{prp:leakage_single}, we show that the channel leakage, indeed, positively affects the ergodic rate in the low-resolution ADC regime, and thus, makes the difference in the ergodic rates of the two channel scenarios.}
%In other words, with the proposed scheduling algorithm for a single propagation path case, there 

%%%%%%%%%%%%%%%%%%%%%%%%%%%%%%%%%%%
%\section{Extension to ($N_{\rm RF} < N_r$)}
%%%%%%%%%%%%%%%%%%%%%%%%%%%%%%%%%%%
%
%%%%%%%%%%%%%%%%%%%%%%%%%%%%%%%%%%%
%\subsection{User and Beam Selection}
%%%%%%%%%%%%%%%%%%%%%%%%%%%%%%%%%%%
%
%
%In hybrid beamforming architectures, beam selection needs to be considered as well as user selection. 
%
%\begin{align}
%	\mathcal{R} =\max_{\substack{\mathcal{S} \subset \{1,\dots,N_u\}:|\mathcal{S}| \leq N_{\rm RF}\\ \mathcal{B} \subset \{1,\dots,N_r\} :|\mathcal{B}| = N_{\rm RF}}} {\sum_{k \in \mathcal{S}} \mathcal{R}_k}
%	%  \log_2 \left(1+\frac{\alpha \rho}{(1-\alpha){\bf w}_{{\rm zf},k}^H {\rm diag}\Big(\rho{\bf H}_{\rm b} {\bf H}_{\rm b}^H + \frac{1}{1-\alpha}{\bf I}_{N_{\rm RF}}\Big) {\bf w}_{{\rm zf},k}}\right)  
%\end{align}

%%%%%%%%%%%%%%%%%%%%%%%%%%%%%%%%%%
\section{Simulation Results}
\label{sec:simulation}
%%%%%%%%%%%%%%%%%%%%%%%%%%%%%%%%%%

In this section, we numerically evaluate the proposed algorithms, validate the derived ergodic rates, and confirm intuitions in this paper. 
In simulations, the number of channel paths $L_k$ is distributed as $L_k \sim \max\{{\rm Poission}(\lambda_L),1\}$ \cite{akdeniz2014millimeter} where $\lambda_L$ represents the near average number of channel paths.
We consider $M = 128$ BS antennas
% and $N =40$ RF chains which is about $30 \%$ of the number of BS antennas $M$. 
and $K = 200$ candidate users, and the BS schedules $S = 12$ users to serve at each transmission \cite{malkowsky2017world,vieira2014flexible}.
%yang2017design
Without imposing the constraint of $\|{\bf h}_{{\rm b},k}\| = \sqrt{\gamma_k}$, the following cases are evaluated through simulation: (1) CSS algorithm, (2) greedy algorithm, (3) chordal distance-based algorithm, (4) mmWave beam aggregation-based scheduling (mBAS) algorithm \cite{lee2016performance}, and (5) SUS algorithm \cite{yoo2006optimality}.
To provide a reference for a performance lower bound, a random scheduling case is also included.
%, which randomly schedules $N_s$ users.
%  the channel model \eqref{eq:channel_geo} is used in simulations.
%with the complex gains $\omega_{k,\ell}\sim \cC\cN(0,1)$ and the angle of arrivals $\theta_{k,\ell}\sim {\rm Unif}(-\pi/2,\pi/2)$ 
%In simulations, complex gains of channel propagation paths are generated by following $\cC\cN (0,1)$ without the constraint of $\|{\bf h}_{{\rm b},k}\| = \sqrt{\gamma}$, $\forall k$.
For the CSS and the mBAS algorithms, the BS stores $N_b = L_k$ indices of dominant elements in the effective channel ${\bf h}_{{\rm b},k}$.
Parameters such as $\epsilon_{th}$, $N_{\rm OL}$, and $d_{th}$ are optimally chosen unless mentioned otherwise.

%%%%%%%%%%%%%%%%%%%%%%%%%%%%%%%%%%
\subsection{Performance Validation}
\label{sec:sim_full}
%%%%%%%%%%%%%%%%%%%%%%%%%%%%%%%%%%

% FIGURE %%%%%%%%%%%%%%%%%%%%%%%%%
\begin{figure}[t]
\centering
$\begin{array}{c c}
{\resizebox{0.5\columnwidth}{!}
{\includegraphics{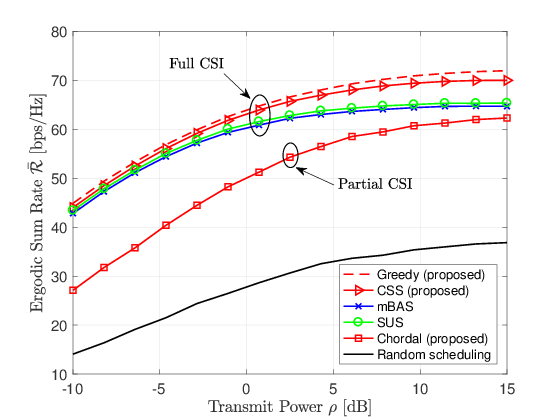}}
}&
{\resizebox{0.49\columnwidth}{!}
{\includegraphics{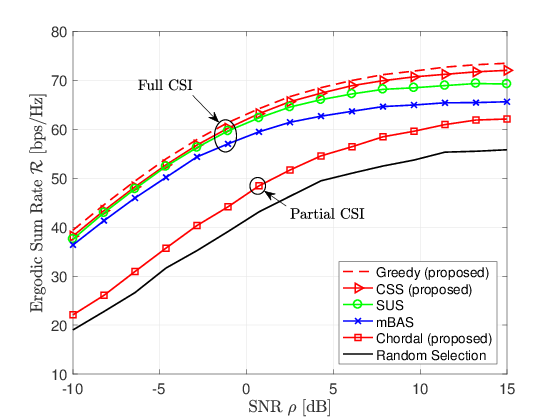}}
}\\
\mbox{\small (a) $\lambda_L = 3$} & \mbox{\small (b) $\lambda_L = 8$}
\end{array}$
\caption{
\color{black} Uplink sum rate simulation results for $M = 128$ BS antennas, $N = 40$ RF chain, $K = 200$ candidate users, $S = 12$ scheduled users, and $b =3$ quantization bits with (a) $\lambda_L = 3$ and (b) $\lambda_L = 8$ average channel paths,.} 
\label{fig:algorithms}
\end{figure}
%%%%%%%%%%%%%%%%%%%%%%%%%%%%%%%%%%

We first focus on performance validation of the proposed algorithms in sum rate.
In Fig. \ref{fig:algorithms}, we consider  $N = 40$ RF chains which is about $30\%$ of the number of antennas $M=128$ and  $b = 3$ quantization bits.
Fig. \ref{fig:algorithms}(a) shows the uplink sum rate with respect to the SNR $\rho$ for $\lambda_L = 3$.
% At each transmission, the $N$ RF chains are randomly selected.
We note that the proposed CSS algorithm achieves the higher sum rate compared to the SUS and mBAS algorithms. 
In addition, the CSS algorithm attains the sum rate that is comparable to that of the proposed greedy algorithm which achieves the sub-optimal rate by requiring much higher complexity.
The sum rate gap between the CSS and the prior algorithms---the SUS and mBAS algorithms---increases as $\rho$ increases because the quantization noise becomes dominant compared to the AWGN in the high SNR regime.
%The performance gap between the CSS algorithm and the prior algorithms $\--$ SUS and mBAS $\--$ increases as the transmit power $\rho$ increases.
%Since the AWGN is dominant when the transmit power is low, the noise determines the power of quantization error.
%Accordingly, in the low SNR regime, the approximated SINR in \eqref{eq:sinr} reduces to the channel power $\|{\bf h}_{{\rm b},k}\|^2$ which is the scheduling measure of the SUS and mBAS, and thus, the algorithms show similar performance. 

{\color{black}
Fig.~\ref{fig:algorithms}(b) plots simulation results with $\lambda_L\! =\! 8$ average channel paths for $\sum_{k=1}^S L_S(k) > N$ where the condition in Theorem 1 does not hold. 
The proposed CSS algorithm achieves a higher sum rate than conventional scheduling methods, which shows that although the derived scheduling criteria may not be optimal in a practical system, they can still be effective for mmWave user scheduling as they capture a relationship between the sparse property of mmWave channels and quantization error. 
In Fig.~\ref{fig:algorithms}(a) and (b), the chordal distance-based algorithm which only exploits the AoA knowledge improves the sum rate compared to random scheduling, closing the gap between the SUS and mBAS algorithms.
% The chordal distance-based algorithm also provides the sum rate improvement compared to the random case.
Therefore, the simulation results validate the sum rate performance of the proposed algorithms.
}

% % FIGURE %%%%%%%%%%%%%%%%%%%%%%%%%
% \begin{figure}[!t]\centering
% \includegraphics[scale = 0.43]{Algorithms_bits.png}
% \caption{Uplink sum rate with respect to quantization bits $b$ for $M = 128$ BS antennas,  $N = 128$ RF chains, $K = 200$ candidate users, $S = 12$ scheduled users, $\lambda_L = 3$ average channel paths, and $\rho = 6$ dB transmit power.} 
% \label{fig:bits}
% \end{figure}
% %%%%%%%%%%%%%%%%%%%%%%%%%%%%%%%%%%

% FIGURE %%%%%%%%%%%%%%%%%%%%%%%%%
\begin{figure}[t]
\centering
$\begin{array}{c c}
{\resizebox{0.5\columnwidth}{!}
{\includegraphics{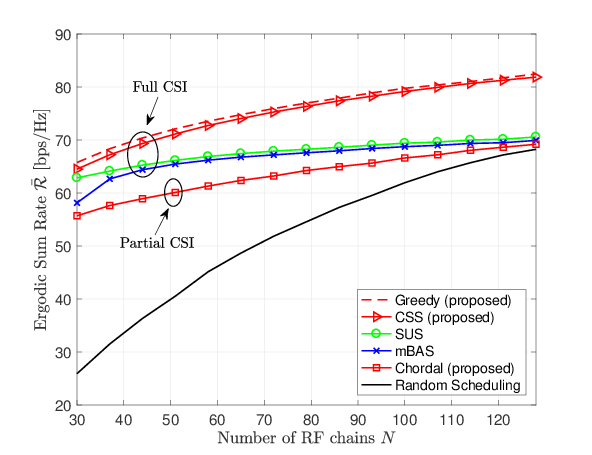}}
}&
{\resizebox{0.51\columnwidth}{!}
{\includegraphics{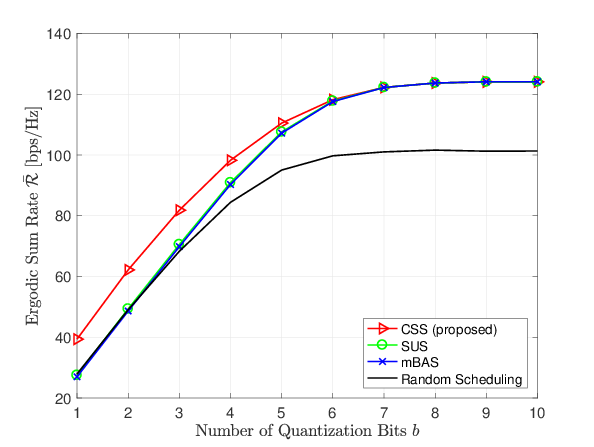}}
}\\
\mbox{\small (a)} & \mbox{\small (b)}
\end{array}$
\caption{
Uplink sum rate for $M = 128$ antennas, $K = 200$ candidate users, $S = 12$ scheduled users, $\lambda_L = 3$ average channel paths, and $\rho = 6$ dB SNR with respect to the number of (a) RF chains $N$ with $b =3$ and (b) quantization bits $b$ with  $N = 128$.} 
\label{fig:algorithms2}
\end{figure}
%%%%%%%%%%%%%%%%%%%%%%%%%%%%%%%%%%

In Fig. \ref{fig:algorithms2}(a), the sum rate results with respect to the number of RF chains $N$ are presented for $\rho = 6$ dB. 
The CCS algorithm shows its sum rate that tightly aligns with that of the greedy algorithm, achieving the higher rate than the SUS and mBAS.
In addition, the chordal distance-based algorithm shows a large improvement compared to the random scheduling for the low to medium $N$.
%The chordal distance-based algorithm also discloses the decreasing gap from the SUS and mBAS as $N$ increases. 
As $N$ increases, the effective channels ${\bf h}_{{\rm b},k}$ are more likely to be orthogonal to each other for the fixed number of scheduled users, which enhances the performance of random scheduling. 
In this regard, the sum rates of the SUS and mBAS algorithms show the marginal sum rate increase compared to the random scheduling as $N$ increases, whereas the CSS algorithm still provides the noticeable improvement by mitigating quantization error.
%This indicates that the CCS algorithm which considers quantization error in scheduling becomes more effective compared to the SUS and mBAS algorithms as $N$ increases due to the increasing number of ADCs.  
% Therefore, we demonstrate the gain from using the derived scheduling criteria in user scheduling.
% in Fig. \ref{fig:algorithms2}.
%In the high SNR regime, the proposed CSS algorithm provides $22 \%$ sum rate increase compared to the random scheduling, whereas the algorithm in \cite{lee2016performance} and the SUS algorithm show marginal sum rate increase.

% In Fig. \ref{fig:algorithms2}(b), we assume $\lambda_L = 3$ average number of propagation paths, $N = 128$ RF chains, and $\rho = 6$ dB transmit power.
Fig. \ref{fig:algorithms2}(b) shows the uplink sum rate with respect to the number of quantization bits $b$.
The CSS algorithm also attains the sum rate of the greedy algorithm with lower complexity and outperforms the SUS and mBAS algorithms.
Note that the sum rate of the SUS and mBAS algorithms converges to that of the CSS and greedy algorithms as the number of quantization bits $b$ increases; i.e., quantization error becomes negligible.
This convergence corresponds to the fact that the derived criteria is effective under coarse quantization.
Thus, in the low-resolution ADC regime, the CSS algorithm provides the noticeable sum rate improvement compared to the other algorithms that ignore quantization error.
%Consequently, in the low-resolution regime in particular, there is a noticeable sum rate gap between the CSS algorithm and the other algorithms that ignore quantization error in user scheduling.

% % FIGURE %%%%%%%%%%%%%%%%%%%%%%%%%
% \begin{figure}[!t]\centering
% \includegraphics[scale = 0.35]{Setsize.png}
% \caption{The number of users in the candidate set at each stage $i$ for $\rho = 6$ dB transmit power, $M = 128$ BS antennas, $N = 128$ RF chains, $K = 200$ candidate users, $S = 12$ scheduled users, $\lambda_L = 3$ average channel paths, and $b =3$ quantization bits.} 
% \label{fig:setsize}
% \end{figure}
% %%%%%%%%%%%%%%%%%%%%%%%%%%%%%%%%%%

% The average number of remaining users in the candidate set at each scheduling stage $i$ is shown in Fig. \ref{fig:setsize}.
% The simulation result is corresponding to the $\rho = 6$ dB point in Fig. \ref{fig:algorithms}(a). 
% The CSS, SUS, mBAS, and chordal distance-based algorithm achieve large decrease in size of the candidate set with scheduling stages owing to the set filtering.
% % of candidate users by using orthogonality conditions.
% In particular, the proposed CSS attains the similar size reduction with the SUS while achieving the higher sum rate that is comparable to the greedy scheduling.
% Consequently, the simulation results reveal the efficiency of the CSS algorithm which accomplishes the high sum rate performance with low complexity regarding the number of remaining users in the candidate set.

%%%%%%%%%%%%%%%%%%%%%%%%%%%%%%%%%%
\subsection{Analysis Validation}
\label{sec:sim_partial}
%%%%%%%%%%%%%%%%%%%%%%%%%%%%%%%%%%

% FIGURE %%%%%%%%%%%%%%%%%%%%%%%%%
\begin{figure}[t]
\centering
$\begin{array}{c c}
{\resizebox{0.5\columnwidth}{!}
{\includegraphics{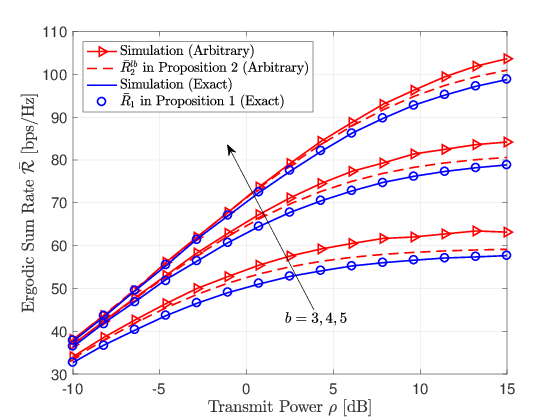}}
}&
{\resizebox{0.5\columnwidth}{!}
{\includegraphics{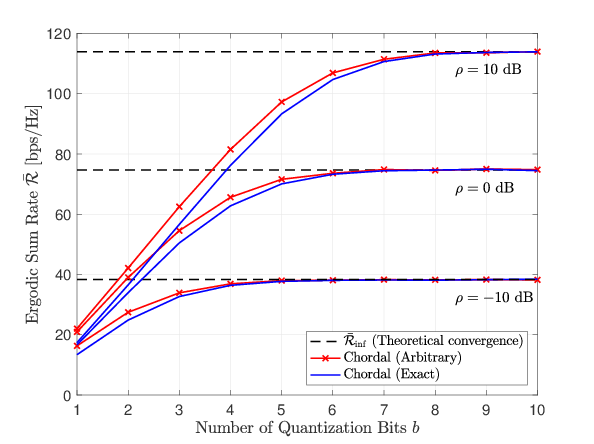}}
}\\
\mbox{(a)} & \mbox{(b)}
\end{array}$
\caption{
(a) The analytical and simulation results for the uplink sum rate of the system with chordal distance-based scheduling, and (b) simulation results for the uplink sum rate of the system with chordal distance-based scheduling for $M = 128$ BS antennas, $N = 128$ RF chains, $K = 200$ candidate users, $S = 12$ scheduled users, and $L_k = 1$ channel path $\forall k$, } 
% with $b \in \{3,4,5\}$ quantization bits,  with $\rho \in \{-10, 0, 10\} $ dB transmit power
\label{fig:Propositions}
\end{figure}
%%%%%%%%%%%%%%%%%%%%%%%%%%%%%%%%%%

In this subsection, we validate the performance analysis and intuitions obtained from the analyses.
In Fig. \ref{fig:Propositions}, we consider $N = 128$  and $L_k = 1$, $\forall k$.
As shown in Fig. \ref{fig:Propositions}(a), the derived ergodic rate \eqref{eq:ergodicL1} in Proposition \ref{prp:ergodicL1} exactly matches the ergodic rate from the simulation.
In addition, the lower bound approximation of ergodic rate \eqref{eq:leakage_single} in Proposition \ref{prp:leakage_single} shows a small gap from the ergodic rate of the simulation, validating its analytical accuracy.
In this regard, the derived ergodic rates can provide a performance guideline for the hybrid MIMO systems with the proposed chordal distance-based algorithm.
From Fig. \ref{fig:Propositions}(a), we note that the two different channel scenarios---exact AoA alignment and arbitrary AoAs---show difference in sum rate for the same system configuration, as discussed in Remark \ref{rm:bit_increase}.
In the following simulation results, we numerically examine this phenomenon based on intuitions obtained in this paper.

%%%%%%%%%%%%%%%%%% FIGURE %%%%%%%%%%%%%%%%%
%\begin{figure}[!t]\centering
%\includegraphics[scale = 0.48]{Proposition_txpower}
%\caption{The analytical approximations and simulation results for the uplink sum rate of the system with chordal distance-based scheduling with $M = 128$ BS antennas, $N = 128$ RF chains, $K = 200$ candidate users, $S = 12$ scheduled users, $L_k = 1$ channel path $\forall k$, and $b \in \{3,4,5\}$ quantization bits.} 
%\label{fig:proposition}
%\end{figure}
%%%%%%%%%%%%%%%%%%%%%%%%%%%%%%%%%%%%%%%%%%%

%%%%%%%%%%%%%%%%%%% FIGURE %%%%%%%%%%%%%%%%%
%\begin{figure}[!t]\centering
%\includegraphics[scale = 0.45]{Remark4_bit.png}
%\caption{Simulation results for the uplink sum rate of the system with chordal distance-based scheduling with $M = 128$ BS antennas, $N = 128$ RF chains, $K = 200$ candidate users, $S = 12$ scheduled users, $L_k = 1$ channel path $\forall k$, and $\rho \in \{-10, 0, 10\} $ dB transmit power.} 
%\label{fig:remark}
%\end{figure}
%%%%%%%%%%%%%%%%%%%%%%%%%%%%%%%%%%%%%%%%%%%%

We evaluate the sum rate of the chordal distance-based scheduling algorithm with respect to the number of quantization bits $b$ to find the behavior of the sum rate gap between the two channel scenarios: exact AoA alignment and arbitrary AoAs.
%We consider $N = 128$ RF chains, $L_k = 1$ channel paths for all users and $\rho \in \{-10,0,10\}$ dB transmit power.
In Fig. \ref{fig:Propositions}(b), it is shown that the uplink sum rates converges to $\bar\cR_{\rm inf} = \frac{S}{\ln 2}e^\frac{1}{\rho M}\,\Gamma\left(0,\frac{1}{\rho M}\right)$  as $b$ increases.
As discussed in Remark~\ref{rm:bit_increase}, such convergence of the sum rates implies that the two channel scenarios lead to different effects on quantization error.
We can also note that the convergence rates are different for different $\rho$.
When the SNR is low, the quantization noise is less dominant compared to the AWGN, which results in faster convergence in terms of the number of $b$, and vice versa.
Therefore, we can conclude that coarse quantization causes the different sum rates from the channel scenarios.

%%%%%%%%%%%%%%%%%% FIGURE %%%%%%%%%%%%%%%%%
\begin{figure}[!t]\centering
\includegraphics[scale = 0.45]{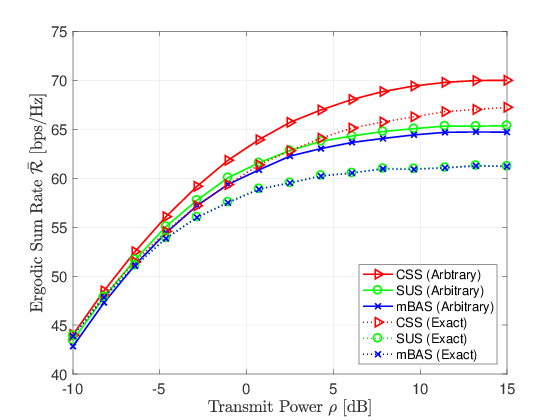}
\caption{Uplink sum rate simulation results for $M = 128$ BS antennas, $N = 40$ RF chains, $K = 200$ candidate users, $S = 12$ scheduled users, $\lambda_L = 3$ average channel paths, and $b = 3$ quantization bits.} 
\label{fig:intuition}
\vspace{-1 em}
\end{figure}
%%%%%%%%%%%%%%%%%%%%%%%%%%%%%%%%%%%%%%%%%%%

In Fig. \ref{fig:intuition}, we simulate the sum rates for the two channel scenarios
% to examine the effect of the channels under coarse quantization.
with $N = 40$, $\lambda_L = 3$, and $b = 3$.
We note that the sum rate for the arbitrary AoA channel is higher than that for the exact AoA alignment channel in the medium and high SNR regime in which the quantization noise is dominant over the AWGN.
{\color{black}
The quantization noise variance at the $i$th ADC is computed as $\bbE[|y_i - y_{{\rm q},i}|^2] = \frac{\pi\sqrt{3}}{2}\sigma_i^2 2^{-2b}$ \cite{orhan2015low}, where $\sigma_i^2 = \bbE[|y_i|^2] = p_u\|[\bH_{\rm b}]_{i,:}\|^2 + 1$. 
Therefore, without the phase offset, most $\sigma_i^2$ would be large whereas most $\sigma_i^2$ would be moderate with the phase offsets as the phase offsets spread the channel path gain at certain angles over the entire angles of RF chains.
Consequently, the phase offset reduces the overall quantization noise variance and this leads to the performance gain.
This corresponds to the results in Theorem 1-(ii), i.e., it is more beneficial to have more spread beamspace gains than to have concentrated beamspace gains.}
% The increase in sum rate corresponds to the finding in Theorem \ref{thm:rate_multi} that channel power needs to evenly spread to maximize the sum rate, i.e., the phase offsets from the arbitrary AoAs result in channel leakage, spreading the channel power.
% This phenomenon reduces the quantization error and achieves sum rate gain.
% Consequently, we show that the channel leakage positively affects the ergodic rate in the hybrid system with low-resolution ADCs.

%%%%%%%%%%%%%%%%%%%%%%%%%%%%%%%%%%
\section{Conclusion}
%%%%%%%%%%%%%%%%%%%%%%%%%%%%%%%%%%

This paper investigated user scheduling for mmWave hybrid beamforming systems with low-resolution ADCs. 
We proposed new user scheduling criteria that are effective under coarse quantization.
%; to maximize the achievable rate, the channels of the scheduled users need to have $(1)$ as many propagation paths as possible with unique AoAs to give spatial orthogonality in the beamspace, and $(2)$ even power distribution in the beamspace to reduce the quantization error.
Leveraging the criteria, we developed the user scheduling algorithm which achieves the sub-optimal sum rate with low complexity, outperforming the conventional scheduling algorithms.
We further proposed the chordal distance-based scheduling algorithm which only exploits the AoA knowledge of channels. 
%Using the AoA which is the slowly-varying channel characteristics, the proposed algorithm can greatly reduce the burden of estimating instantaneous full CSI at each channel coherence time.
The chordal distance-based scheduling algorithm improved the sum rate compared to the random scheduling case, closing the gap between the full CSI-based conventional scheduling methods as the SNR increases.
We also provided the performance analysis for the algorithm in ergodice rate, and the derived rates are the functions of system parameters including quantization bits. 
We obtained an intuition from the derived rates that channel leakage due to the phase offsets between the arbitrary AoAs and quantized angles of analog combiners offers the sum rate gain by reducing the quantization error compared to the channel without leakage.
This intuition challenges the conventional negative understanding towards channel leakage and is validated through simulation.
Therefore, for mmWave communications, this paper provides not only new user scheduling algorithms for low-resolution ADC systems, but also new scheduling criteria and intuition for mmWave channels under coarse quantization.
{\color{black} For potential future work, investigating user scheduling for sub-array based hybrid systems with low-resolution ADCs is desirable to consider more practical receiver architectures.}

%%%%%%%%%%%%%%%%%%
\begin{appendices}
\section{Proof of Proposition 1}
\label{appx:proposition1}
%%%%%%%%%%%%%%%%%%%%%%%%%%%%

Let the ZF combiner ${\bf W}_{\rm zf} = {\bf H}_{\rm b}(\cS_{\rm cd})({\bf H}_{\rm b}(\cS_{\rm cd})^H {\bf H}_{\rm b}(\cS_{\rm cd}))^{-1}$.
Using the achievable rate \eqref{eq:rate}, the ergodic rate of user $k \in \cS_{\rm cd}$ is defined as
\begin{align}
	\label{eq:ergodic}
	\bar{r}_k = \bbE\Big[r_k\big(\bH_{\rm b}(\cS_{\rm cd})\big)\Big] 
	& = \mathbb{E}\left[ \log_2 \left(1+\frac{\alpha^2 \rho}{ {\bf w}_{{\rm zf},k}^H {\bf R}_{{\bf qq}}(\bH_{\rm b}(\cS_{\rm cd})){\bf w}_{{\rm zf},k}+\alpha^2 \|{\bf w}_{{\rm zf},k}\|^2}\right) \right].
%	& = \mathbb{E}\left[ \log_2 \Bigg(1+\frac{\alpha \rho}{(1-\alpha){\bf w}_{{\rm zf},k}^H {\rm diag}\Big(\rho{\bf H}_{\rm b} {\bf H}_{\rm b}^H + \frac{1}{1-\alpha}{\bf I}_{N_{\rm RF}}\Big) {\bf w}_{{\rm zf},k}}\Bigg) \right]\\
\end{align}
Based on Remark \ref{rm:chordal_L1}, the algorithm schedules a user $j \in \cK_{\rm cd}$ who provides the smallest value of $|{\bf a}^H(\phi_k) {\bf a}(\phi_{j})|$.
Under the assumption of the exact AoA alignment, $|{\bf a}^H(\phi_k) {\bf a}(\phi_{j})|$ is equivalent to zero when $\mathcal{L}_{k} \cap \mathcal{L}_{j} = \emptyset$ for $ k \neq j$, i.e., user channels are spatially orthogonal to each other.
%The probability of unique AoAs between scheduled users,  $P(\mathcal{L}_{k} \cap \mathcal{L}_{k'} = \phi)$, $k, k' \in \cS_{\rm cd}$, goes to $1$ as the number of antennas $M$ and candidate users $K$ increase with the fixed number of channel paths $L$. 
%In the large antenna array regime with the large number of users, we have $P(\mathcal{L}_{k} \cap \mathcal{L}_{k'} = \phi) \approx 1$ for $L=1$. 
For the exact AoA alignment scenario with $L =1$, there is only one non-zero element in $\bh_{{\rm b},k}$. 
Accordingly, any scheduled users have to satisfy $\mathcal{L}_{k} \cap \mathcal{L}_{j} = \emptyset$ to avoid rank deficiency of a channel matrix, which can be guaranteed by setting $|{\bf a}^H(\phi_k) {\bf a}(\phi_{k'})| < \epsilon_{th} \ll 1$ in the filtering.
% which forces semi-orthogonality between scheduled users.
Hence, the ZF combiner for user $k \in \cS_{\rm cd}$ becomes ${\bf w}_{{\rm zf},k} = {\bf h}_{{\rm b}, k}/\|{\bf h}_{{\rm b}, k}\|^2$, and \eqref{eq:ergodic} is solved as
%	The optimal user scheduling without knowing channel gains is to schedule users with unique AoAs, i.e.,  $\mathcal{L}_{k} \cap \mathcal{L}_{k'} = \phi$ for $ k \neq k'$ as shown Theorem \ref{thm:rate_multi} so that user channels are spatially orthogonal to each other ${\bf h}_{{\rm b},k}  \perp {\bf h}_{{\rm b},k'}$.
%	Accordingly, the ZF combiner becomes ${\bf w}_{{\rm zf},k} = {\bf h}_{{\rm b}, k}/\|{\bf h}_{{\rm b}, k}\|^2$, and \eqref{eq:ergodic} is
\begin{align}
	\label{eq:ergodicL1_pf}
	\bar{r}_k  
	& = \bbE\left[\log_2 \Biggl(1+\frac{\alpha \rho {\|{\bf h}_{{\rm b},k}\|^4} }{\rho (1-\alpha){\bf h}^H_{{\rm b},k}{\rm diag}\Big({\bf H}_{\rm b}(\cS_{\rm cd}) {\bf H}_{\rm b}(\cS_{\rm cd})^H\Big) {\bf h}_{{\rm b},k}+ \|{\bf h}_{{\rm b},k}\|^2}\Biggr) \right] \\ \nonumber
%	& \approx \mathbb{E}\left[ \log_2 \Bigg(1+\frac{\alpha \rho}{(1-\alpha) \frac{{\bf h}_{{\rm b}, k}^H } {\|{\bf h}_{{\rm b}, k}\|^2} {\rm diag}\Big(\rho{\bf H}_{\rm b} {\bf H}_{\rm b}^H + \frac{1}{1-\alpha}{\bf I}_{N_{\rm RF}}\Big)  \frac{{\bf h}_{{\rm b}, k} } {\|{\bf h}_{{\rm b}, k}\|^2}}\Bigg) \right] 
%		\\ \nonumber
	& \stackrel{(a)} = \mathbb{E}\left[ \log_2 \Bigg(1+\frac{\alpha \rho}{(1-\alpha)\rho + 1/(M|g_{k}|^2)}\Bigg) \right]
		\\  \label{eq:ergodicL1_pf1}
	&\stackrel{(b)}= \frac{1}{\ln 2}\left(e^\frac{1}{\rho M}\,\Gamma\left(0,\frac{1}{\rho M}\right)-e^\frac{1}{\rho (1-\alpha) M }\,\Gamma\left(0,\frac{1}{\rho (1-\alpha)M}\right)\right)
%		&  \stackrel{(b)}\leq \log_2\Bigg(1 + \alpha \rho\mathbb{E}\bigg[\frac{1}{\rho(1-\alpha)+ 1/(M|h_{{\rm b},k,i}|^2)}\bigg]\Bigg)
%		\\  \label{eq:ergodicL1_pf2}
%		&  \stackrel{(c)}= \log_2\Bigg(1 + \alpha \rho\left(\frac{1}{\rho(1-\alpha)} - \frac{e^{\frac{1}{M\rho (1-\alpha)}}}{M\rho^2(1-\alpha)^2}\Gamma\left(0, \frac{1}{M\rho(1-\alpha)}\right)\right)\Bigg)
%		&  \stackrel{(c)}= \log_2\Bigg(1 + \alpha \rho\frac{1}{\rho(1-\alpha)} - \frac{e^{\frac{1}{\rho (1-\alpha)}}}{\rho^2(1-\alpha)^2}\Gamma\left(0, \frac{1}{\rho(1-\alpha)}\right)\Bigg)
\end{align}
where $g_k$ is the complex gain of the propagation path of user $k$.
Here, (a) is from $L = 1$ with $\mathcal{L}_{k} \cap \mathcal{L}_{k'} = \emptyset$ for $k, k' \in \cS_{\rm cd}$, and (b) comes from the fact that $|g_{k}|^2$ is an exponential random variable with the rate parameter $\lambda = 1$, $|g_{k}|^2 \sim {\rm Exp}(1)$.
%Note that \eqref{eq:ergodicL1_pf} is not an upper bound of $\bar{r}_k$ even with ${\bf w}_{{\rm zf},k} = {\bf h}_{{\rm b}, k}/\|{\bf h}_{{\rm b}, k}\|^2$.
%In the no quantization case, ${\rm diag}(\rho{\bf H}_{\rm b}(\cS_{\rm cd}){\bf H}_{\rm b}(\cS_{\rm cd})^H)$ does not exist and \eqref{eq:ergodicL1_pf} becomes an upper bound as ${\bf w}_{{\rm zf},k} = {\bf h}_{{\rm b}, k}/\|{\bf h}_{{\rm b}, k}\|^2$ is the optimal ZF combiner.
%In the coarse quantization case, however, it is not guaranteed to be optimal as the quantization error term ${\rm diag}({\bf H}_{\rm b}(\cS_{\rm cd}){\bf H}_{\rm b}(\cS_{\rm cd})^H)$ needs to be jointly minimized.
Due to the randomness of $g_k$, the ergodic rate of each user is equal, which leads to \eqref{eq:ergodicL1}.
This completes the proof.
%	Using $\mathbb{E}[\frac{1}{X}] = \mathbb{E}\left[\int_{0}^{\infty} e^{-sX} ds\right]$, the 	expectation in \eqref{eq:ergodicL1_pf1} becomes
%	\begin{align}
%		\nonumber
%		\mathbb{E}\bigg[\frac{1}{\rho(1-\alpha)+ 1/|h_{b,k,i}|^2}\bigg] & = \mathbb{E}\left[\int_{0}^{\infty} e^{-s\left({\rho(1-\alpha) + \frac{1}{|h_{{\rm b},k,i}|^2}}\right)} ds\right] 
%		\\ \nonumber
%		& = \int_{0}^{\infty} e^{-s\, \rho(1-\alpha)} \mathbb{E}\left[e^{-\frac{s}{|h_{{\rm b},k,i}|^2}}\right] ds
%		\\ \nonumber
%		& \stackrel{(a)}= \int_{0}^{\infty} e^{-s\, \rho(1-\alpha)} K_1(2\sqrt{s})\, ds
%		\\ \label{eq:ergodicL1_pf2}
%		& =\frac{1}{\rho(1-\alpha)} - \frac{e^{\frac{1}{\rho (1-\alpha)}}}{\rho^2(1-\alpha)^2}\Gamma\left(0, \frac{1}{\rho(1-\alpha)}\right)
%	\end{align}
%	where (c) comes from that $|h_{b,k,i}|^2$ is an exponential random variable with the rate parameter $\lambda = 1$, $|h_{b,k,i}|^2 \sim {\rm Exp}(1)$, and $K_n(z)$ is a modified Bessel function of the second kind.
%	Combining \eqref{eq:ergodicL1_pf2} with \eqref{eq:ergodicL1_pf1} and considering $N_s$ users, we derive the final result \eqref{eq:ergodicL1} in Proposition \ref{prp:ergodicL1}.
\qed

%%%%%%%%%%%%%%%%%%%%%%%%%%%%
\section{Proof of Proposition 2}
\label{appx:proposition2}
%%%%%%%%%%%%%%%%%%%%%%%%%%%%

%Let $C_\perp$ be the condition that the array response vectors of the scheduled users are orthogonal, i.e., ${\bf a}(\phi_k) \perp {\bf a}(\phi_{k'})$, $\forall k, k' \in \cS_{\rm cd}$. 
To find a lower bound of the ergodic sum rate achieved by the proposed algorithm, we consider the random scheduling method and find its ergodic sum rate for the lower bound. 
%For this proof,  we utilize spatial angles $\vartheta$ and $\varphi$ instead of physical angles $\theta$ and $\phi$.    
Since we focus on a large antenna array system at the BS, the array response vectors of the scheduled users are almost orthogonal with large $M$ \cite{ngo2014aspects}, and thus we adopt ${\bf w}_{{\rm zf},k} \approx \frac{{\bf A}^H{\bf h}_{ k}}{\|{\bf h}_k\|^2}$. Then, the ergodic rate of the scheduled user $k$ can be approximated as 
%According to Remark \ref{rm:chordal_L1}, the chordal distance-based algorithm schedules a user $j \in \cK_{\rm cd}$ who provides the smallest value of $|{\bf a}^H(\phi_k) {\bf a}(\phi_{j})|$ satisfying $|{\bf a}^H(\phi_k) {\bf a}(\phi_{k'})| < \epsilon_{\rm th}\ll 1$.
%Since we schedule user by finding users whose array response vectors are as orthogonal as possible, the optimal scheduling case only with AoAs information would be ${\bf a}(u_i) \perp {\bf a}(u_j)$, $\forall i \neq j$.
%Then, the ergodic rate of the user $k$ can be approximated as
\begin{align}
	\nonumber
	&\bar{r}_k  = \mathbb{E}\left[ \log_2 \left(1+\frac{\alpha^2 \rho}{ {\bf w}_{{\rm zf},k}^H {\bf R}_{{\bf qq}}\big(\bH_{\rm b}(\cS_{\rm cd})\big){\bf w}_{{\rm zf},k}+\alpha^2 \|{\bf w}_{{\rm zf},k}\|^2}\right)\right]\\ \label{eq:leakage_proof}
	&\stackrel{(a)}  \approx  \mathbb{E}\left[\log_2\left(1+\frac{\alpha\rho\|{\bf h}_k\|^4}{(1-\alpha)({\bf A}^H{\bf h}_k)^H{\rm diag}\big(\rho {\bf A}^H{\bf H}(\cS_{\rm cd}){\bf H}^H(\cS_{\rm cd}){\bf A}\big){\bf A}^H{\bf h}_k + \|{\bf h}_k\|^2}\right)\right],
%	&\stackrel{(a)}  =  \mathbb{E}\left[\log_2\left(1+\frac{\alpha\rho\|{\bf h}_k\|^4}{(1-\alpha)({\bf A}^H{\bf h}_k)^H{\rm diag}(\rho {\bf A}^H{\bf H}(\cS_{\rm cd}){\bf H}^H(\cS_{\rm cd}){\bf A} + \frac{1}{1-\alpha}{\bf I}){\bf A}^H{\bf h}_k}\right)\bigg|C_\perp\right]
\end{align} 
%Note that \eqref{eq:ergodicL1_pf} is not an upper bound of $\bar{r}_k$ even with ${\bf w}_{{\rm zf},k} = {\bf h}_{{\rm b}, k}/\|{\bf h}_{{\rm b}, k}\|^2$.
%In the no quantization case, ${\rm diag}(\rho{\bf H}_{\rm b}(\cS_{\rm cd}){\bf H}_{\rm b}(\cS_{\rm cd})^H)$ does not exist and \eqref{eq:ergodicL1_pf} becomes an upper bound as ${\bf w}_{{\rm zf},k} = {\bf h}_{{\rm b}, k}/\|{\bf h}_{{\rm b}, k}\|^2$ is the optimal ZF combiner.
%In the coarse quantization case, however, it is not guaranteed to be optimal as the quantization error term ${\rm diag}({\bf H}_{\rm b}(\cS_{\rm cd}){\bf H}_{\rm b}(\cS_{\rm cd})^H)$ needs to be jointly minimized.
where (a) comes from ${\bf w}_{{\rm zf},k} \approx \frac{{\bf A}^H{\bf h}_k}{\|{\bf h}_k\|^2}$. 
Without loss of generality, let $\cS_{\rm cd} = \{1,2,\dots,S\}$.
The channel matrix of scheduled users can be represented as ${\bf H}(\cS_{\rm cd}) = \sqrt{M} \bA_u \bG$ where ${\bf A}_u = [{\bf a}(\varphi_1), \dots, {\bf a}(\varphi_{S})]$ and ${\bG} = {\rm diag}(g_1,\dots, g_{S})$, and \eqref{eq:leakage_proof} becomes
\begin{align}
	\nonumber
	&\mathbb{E}\left[\log_2\left(1+ \frac{M^2\alpha \rho |g_k|^4}{M^2\rho(1-\alpha)|g_k|^2{\bf a}^H(\varphi_k){\bf A}{\rm diag}\big({\bf A}^H{\bf A}_u{\bG} \,{\bG}^H{\bf A}_u^H{\bf A}\big){\bf A}^H{\bf a}(\varphi_k)+M|g_k|^2} \right)\right] 
	\\ \nonumber
	& = \mathbb{E}\left[\log_2\left(1+\frac{M\alpha\rho|g_k|^2}{M\rho(1-\alpha)\sum_{m,s=1}^{M,S}|g_s|^2|{\bf a}^H(\vartheta_m){\bf a}(\varphi_k)|^2|{\bf a}^H(\vartheta_m){\bf a}(\varphi_s)|^2+1}\right)\right] 
	\\ \label{eq:leakage_proof1}
	& = \mathbb{E}_{g_k}\left[\mathbb{E}\left[\log_2\left(1+\frac{M\alpha\rho|g_k|^2}{M\rho(1-\alpha)\sum_{m,s} |g_s|^2|{\bf a}^H(\vartheta_m){\bf a}(\varphi_k)|^2|{\bf a}^H(\vartheta_m){\bf a}(\varphi_s)|^2+1}\right)\!\bigg| g_k\right] \right].
\end{align}
%where and ${\pmb \Omega} = {\rm diag}(\omega_1,\dots, \omega_{N_s})$.
%	\begin{align}
%		 {\bf W}_{{\rm zf}} = {\bf H}_{\rm b}\left({\bf H}_{\rm b}^H{\bf H}_{\rm b}\right)^{-1}
%		{\pmb\Omega}\right)^{-1}
%	\end{align}
%	${\bf a}(u_i)^H {\bf a}(u_j)  \to 0$ for $i \neq j$ as $M \to \infty$ \cite{Marzetta Aspects}, which leads to ${\bf w}_{{\rm zf},k} \approx {\bf A}^H{\bf h}_k/\|{\bf h}_k\|^2$ for large $N_r$.
%	\eqref{eq:leakage_proof1} further reduces to
%	\begin{align}
%		R = & \stackrel{(b)}\approx \log_2\left(1+\frac{\mathbb{E}\left[N_r\alpha\rho|\omega_k|^2 | E\right]}{\mathbb{E}\left[N_r\rho(1-\alpha)\sum_{i=1}^{N_r} \sum_{\ell=1}^{N_s}|\omega_\ell|^2|{\bf a}^H(\theta_i){\bf a}(u_k)|^2|{\bf a}^H(\theta_i){\bf a}(u_\ell)|^2+1 \big | E\right]}\right)
%		\\ \label{eq:leakage_proof2}
%		&\stackrel{(c)} = \log_2\left(1+\frac{N_r\alpha\rho}{N_r\rho(1-\alpha)\mathbb{E}\left[\sum_{i=1}^{N_r} \sum_{\ell=1}^{N_s}|{\bf a}^H(\theta_i){\bf a}(u_k)|^2|{\bf a}^H(\theta_i){\bf a}(u_\ell)|^2\big | E\right]+1 }\right)
%	\end{align}
%	where (b) comes from Lemma 1 in \cite{??} and (c) is from the independence of $\omega_k$ with $u_k$.
To compute the inner expectation in \eqref{eq:leakage_proof1}, we can use Lemma 1 in \cite{hamdi2010useful} as $g_k$ is considered to be a constant given the condition, which makes the signal power and the interference-plus-noise power independent to each other.
Let $\Psi_k =M\rho(1-\alpha)\sum_{m,s=1}^{M,S}|g_s|^2|{\bf a}^H(\vartheta_m){\bf a}(\varphi_k)|^2|{\bf a}^H(\vartheta_m){\bf a}(\varphi_s)|^2$, then the inner expectation in \eqref{eq:leakage_proof1} becomes
\begin{align}
	\nonumber
	\mathbb{E}\left[\log_2\left(1+\frac{M\alpha\rho|g_k|^2}{\Psi_k+1}\right)\!\bigg|g_k\right] & \stackrel{(a)}=\frac{1}{\ln 2} \int_0^{\infty}	\frac{e^{-z}}{z}\left(1-e^{-zM\alpha\rho|g_k|^2}\right)\bbE\left[e^{-z\Psi_k}\Big|g_k\right]dz
	\\ \label{eq:leakage_proof2}
	&\stackrel{(b)}\geq \frac{1}{\ln 2} \int_0^{\infty}	\frac{e^{-z}}{z}\left(1-e^{-zM\alpha\rho|g_k|^2}\right)e^{-z\mathbb{E}[\Psi_k|g_k]}dz
\end{align}
where (a) follows from Lemma 1 in \cite{hamdi2010useful} and (b) comes from Jensen's inequality. To compute the expectation in \eqref{eq:leakage_proof2}, we rewrite it as
\begin{align}
	\nonumber
%		&\mathbb{E}\left[\sum_{i=1}^{N_r} |{\bf a}^H(\theta_i){\bf a}(u_k)|^4 + \sum_{i=1}^{N_r} \sum_{\ell\neq k}^{N_s}|{\bf a}^H(\theta_i){\bf a}(u_k)|^2|{\bf a}^H(\theta_i){\bf a}(u_\ell)|^2\bigg | E\right]
	 \bbE\big[\Psi_k | g_k\big] = M\rho (1-\alpha) \Biggl( &\mathbb{E}\left[\sum_{m=1}^{M}|g_k|^2 \big|{\bf a}^H(\vartheta_m){\bf a}(\varphi_k)\big|^4\bigg|g_k\right] +
	\\ \label{eq:leakage_proof3}
	& \mathbb{E}\left[\sum_{m=1}^{M} \sum_{s\neq k}^{S}|g_s|^2\big|{\bf a}^H(\vartheta_m){\bf a}(\varphi_k)\big|^2\big|{\bf a}^H(\vartheta_m){\bf a}(\varphi_s)\big|^2 \right]\Biggr).
\end{align}
The first expectation term in \eqref{eq:leakage_proof3} can be computed as
\begin{align}
	\mathbb{E}\left[\sum_{m=1}^{M}|g_k|^2 \big|{\bf a}^H(\vartheta_m){\bf a}(\varphi_k)\big|^4\bigg| g_k \right] & = |g_k|^2\sum_{m=1}^{M} \mathbb{E}\left[|{\bf a}^H(\vartheta_m){\bf a}(\varphi_k)\big|^4\right] \stackrel{(a)}= |g_k|^2 M \int_0^1 F^4\left(\delta;M\right)\, d\delta \label{eq:leakage_proof4}
\end{align}
where (a) comes from the fact that $\delta_{m,k}:=\vartheta_m-\varphi_k$ can be regarded as $ \delta_{m,k} \overset{i.i.d.}{\sim} {\rm Unif}\bigl[-1,1\bigr]$ due to the symmetry of the Fej\'{e}r kernel of order $M$, $F(\vartheta;M)$ \cite{strichartz2000way}.
%For the second expectation in \eqref{eq:leakage_proof3}, we use $\mathbb{E}[X] = P(C_\perp)\mathbb{E}[X|C_\perp] + P(C_\perp^c)\mathbb{E}[X|C_\perp^c]$. 
%	For a positive random variable $X$, $\mathbb{E}[X] = P(C_\perp)\mathbb{E}[X|C_\perp] + P(C_\perp^c)\mathbb{E}[X|C_\perp^c]$
%Since we consider a large antenna array, $P(C_\perp) \approx 1$ with large $M$ \cite{ngo2014aspects}.
Then, with $\bbE[|g_s|^2] = 1$, the second expectation term can be expressed as
\begin{align}
	\nonumber
	 &\mathbb{E}\left[\sum_{m=1}^{M} \sum_{s\neq k}^{S}|{\bf a}^H(\vartheta_m){\bf a}(\varphi_k)|^2|{\bf a}^H(\vartheta_m){\bf a}(\varphi_s)|^2\right] 
	   = \sum_{m=1}^{M}\sum_{s\neq k}^{S}\mathbb{E}\Big[|{\bf a}^H(\vartheta_m){\bf a}(\varphi_k)|^2\Big]\mathbb{E}\Big[|{\bf a}^H(\vartheta_m){\bf a}(\varphi_s)|^2\Big] 
	 \\ 
	 & = \sum_{m=1}^{M}\sum_{s\neq k}^{S} \mathbb{E}\left[F^2\left( \delta_{m,k}; M \right)\right]\mathbb{E}\left[F^2\left(\delta_{m,s}; M \right)\right]
	 = (S -1) M \left(\int_0^1 F^2\left(\delta; M \right)\, d\delta  \right)^2.	 \label{eq:leakage_proof5}
\end{align}
%where (a) comes from $\vartheta_k, \vartheta_{s\neq k} \sim {\rm Unif}[0,2]$ and the symmetry of the Fej\'{e}r kernel.
%	From \eqref{eq:leakage_proof4} and, \eqref{eq:leakage_proof3} is approximated as 
%	\begin{align}
%		\bbE\big[\Psi | C_\perp, \omega_k\big] \approx
%		M^2\rho(1-\alpha)\left(|\omega_k|^2 \int_0^1 F^4\left(\theta;M\right)\, d\theta + (S -1)\left(\int_0^1 F^2\left(\theta; M \right)\, d\theta \right)^2\right)
%	\end{align}
Let $c_1 = M\alpha\rho$, $c_2 = M^2\rho(1-\alpha)\int_0^1 F^4\left(\delta;M\right)\, d\delta$, and $c_3 = M^2\rho(1-\alpha)(S-1)\left(\int_0^1 F^2\left(\delta; M \right)\, d\delta \right)^2$. From \eqref{eq:leakage_proof1}, \eqref{eq:leakage_proof2},\eqref{eq:leakage_proof4}, and  \eqref{eq:leakage_proof5}, the ergodic rate $\bar r_k$ is approximately lower bounded by
\begin{align}
	\nonumber
	\bar{r}_k  &\approx  \bbE_{g_k}\left[ \mathbb{E}\left[\log_2\left(1+\frac{c_1|g_k|^2}{\Psi_k+1}\right)\!\bigg| g_k\right]\right]
	 \geq \frac{1}{\ln 2}  \bbE_{g_k} \left[\int_0^{\infty}	\frac{e^{-z}}{z}\left(1-e^{-zc_1|g_k|^2}\right)e^{-z\mathbb{E}[\Psi_k|g_k]}dz\right]
	\\ \nonumber
		%		&\approx \frac{1}{\ln 2}  \bbE_{\omega_k} \left[\int_0^{\infty}	\frac{e^{-z}}{z}\left(1-e^{-zM\alpha\rho|w_k|^2}\right)e^{-zM^2\rho(1-\alpha)\left(|\omega_k|^2 \int_0^1 F^4\left(\theta;M\right)\, d\theta + (S -1)\left(\int_0^1 F^2\left(\theta; M \right)\, d\theta \right)^2\right)}dz\right]
%		\\ \nonumber
	& = \frac{1}{\ln 2}  \int_0^{\infty}	\frac{e^{-(1+c_3)z}}{z}\left(\bbE_{g_k} \left[e^{-c_2z|g_k|^2}\right]-\bbE_{g_k}\left[e^{-(c_1+c_2)z|g_k|^2}\right]\right)dz
%		& = \frac{1}{\ln 2}  \int_0^{\infty}	\frac{e^{-z}}{z}\left(\bbE_{\omega_k} \left[e^{-zM^2\rho(1-\alpha)\mathcal{F}_1|\omega_k|^2}\right]-\bbE_{\omega_k}\left[e^{-z(M\alpha\rho+M^2\rho(1-\alpha)\mathcal{F}_1)|\omega_k|^2}\right]\right)e^{-zM^2\rho(1-\alpha)(S-1)\mathcal{F}_2}dz
		\\ \nonumber
	& \stackrel{(a)} = \frac{1}{\ln 2}  \int_0^{\infty}	\frac{e^{-(1+c_3)z}}{z}\left(\frac{1}{1+c_2z}-\frac{1}{1+(c_1+c_2)z}\right)dz\\
	& = \frac{1}{\ln 2}\left(e^{\frac{1+c_3}{c_1+c_2}}\,\Gamma\left(0, \frac{1+c_3}{c_1+c_2}\right) - e^{\frac{1+c_3}{c_2}}\,\Gamma\left(0, \frac{1+c_3}{c_2}\right) \right)
\end{align}
where (a) comes from the Laplace transform of the exponential distribution $|g_k|^2 \sim \exp(1)$.
Without the fading information of channels, the ergodic rate for each user after the user scheduling is equivalent to each other, which results in \eqref{eq:leakage_single}.
This completes the proof.
% for Proposition \ref{prp:leakage_single}.
\qed
\end{appendices}

%\lesssim \gtrsim
%\clearpage
\bibliographystyle{IEEEtran}
\bibliography{Scheduling_ADC.bib}
\end{document}